\documentclass[11pt]{article}
\usepackage{amsmath, amssymb, amsfonts, amsthm}

\newtheorem{theorem}{Theorem}[section]
\newtheorem{proposition}[theorem]{Proposition}

\newtheorem{lemma}[theorem]{Lemma}

\newtheorem{remark}[theorem]{Remark}

\newcommand{\rd}{{\rm d}}
\newcommand{\be}{\begin{equation}}
\newcommand{\ee}{\end{equation}}
\newcommand{\bey}{\begin{eqnarray}}
\newcommand{\eey}{\end{eqnarray}}

\newcommand{\eps}{\varepsilon}

\newcommand{\bx}{{\bf x}}

\newcommand{\ph}{\varphi}

\newcommand{\cU}{{\cal U}}

\newcommand{\cR}{{\cal R}}
\newcommand{\bR}{{\mathbb R}}
\newcommand{\bC}{{\mathbb C}}
\newcommand{\bN}{{\mathbb N}}

\newcommand{\tr}{\mbox{Tr}}

\newcommand{\wt}{\widetilde}

\newcommand{\const}{\mathrm{const}}

\newcommand{\cF}{{\cal F}}
\newcommand{\cA}{{\cal A}}

\newcommand{\cE}{{\cal E}}

\newcommand{\cH}{{\cal H}}
\newcommand{\cL}{{\cal L}}
\newcommand{\cN}{{\cal N}}

\input epsf

\newcommand{\donothing}[1]{}

\begin{document}
\title{Quantum Fluctuations and Rate of Convergence towards \\ Mean Field Dynamics}
\author{Igor Rodnianski${}^1$\thanks{Partially supported by NSF grant DMS-0702270.}\, and Benjamin Schlein${}^2$\thanks{On leave from the University of
Cambridge. Supported by a Sofja Kovalevskaja Award of the Humboldt Foundation.} \\
\\
Department of Mathematics, Princeton University\\ Princeton, NJ, USA${}^1$ \\ \\
Institute of Mathematics, University of Munich, \\
Theresienstr. 39, D-80333 Munich, Germany${}^2$
\\}

\maketitle

\begin{abstract}
The nonlinear Hartree equation describes the macroscopic dynamics of
initially factorized $N$-boson states, in the limit of large $N$. In
this paper we provide estimates on the rate of convergence of the
microscopic quantum mechanical evolution towards the limiting
Hartree dynamics. More precisely, we prove bounds on the difference
between the one-particle density associated with the solution of the
$N$-body Schr\"odinger equation and the orthogonal projection onto
the solution of the Hartree equation.
\end{abstract}

\section{Introduction}
\setcounter{equation}{0}

We consider an $N$ boson system described on the Hilbert space
$L^2_s (\bR^{3N})$ (the subspace of $L^2 (\bR^{3N})$ consisting of
all functions symmetric with respect to arbitrary permutations of
the $N$ particles) by a mean field Hamiltonian of the form
\begin{equation}\label{eq:ham}
H_N = \sum_{j=1}^N -\Delta_{x_j} + \frac{1}{N} \sum_{i<j}^N V(x_i
-x_j)\,.
\end{equation}
We will specify later assumptions on the interaction potential $V$.
Note the coupling constant $1/N$ in front of the potential energy
which characterizes mean-field models; it makes sure that in the
limit of large $N$ the potential and the kinetic energy are
typically of the same order, and thus can compete to generate
nontrivial effective equation for the macroscopic dynamics of the
system.

\medskip

We consider a factorized initial wave function
\begin{equation}\label{eq:init}
L_s^2 (\bR^{3N}) \ni \; \psi_N (\bx) = \prod_{j=1}^N \ph (x_j) \quad
\text{for some } \ph \in H^1 (\bR^3)\end{equation} with
normalization $\| \ph \|_{L^2 (\bR^3)} = 1$ (so that $\| \psi_N
\|_{L^2 (\bR^{3N})} =1$) and we study its time-evolution
$\psi_{N,t}$, given by the solution of the $N$ body Schr\"odinger
equation
\begin{equation}\label{eq:schr}
i\partial_t \psi_{N,t} = H_N \psi_{N,t} \qquad \text{with initial
data $\psi_{N,0} = \psi_N$.}
\end{equation}
In (\ref{eq:init}) and in what follows we use the notation
$\bx=(x_1,\dots,x_N)\in\bR^{3N}$.

\medskip

Clearly, because of the interaction among the particles, the
factorization of the wave function is not preserved by the time
evolution. However, due to the presence of the small constant $1/N$ in front of
the potential energy in (\ref{eq:ham}), we may expect the total
potential experienced by each particle to be approximated, for large
$N$, by an effective mean field potential, and thus that, in the
limit $N \to \infty$, the solution $\psi_{N,t}$ of (\ref{eq:schr})
is still approximately (and in an appropriate sense) factorized. We
may expect, in other words, that in an appropriate sense
\begin{equation}\label{eq:fac} \psi_{N,t} (\bx) \simeq \prod_{j=1}^N
\ph_t (x_j) \qquad \text{for large $N$.}\end{equation} If
(\ref{eq:fac}) is indeed correct, it is easy to derive a
self-consistent equation for the evolution of the one-particle wave
function $\ph_t$.  In fact, it follows from (\ref{eq:fac}) that
the total potential experienced by a particle at $x$ can be
approximated by the convolution $(V*|\ph_t|^2)(x)$, and thus that
the evolution of the one-particle wave function $\ph_t$ is described
by the nonlinear Hartree equation
\begin{equation}\label{eq:hartree}
i \partial_t \ph_t = -\Delta \ph_t + (V
* |\ph_t|^2) \ph_t \,. \end{equation}

\medskip

To understand in which sense (\ref{eq:fac}) holds true, we need to
introduce marginal densities. The density matrix $\gamma_{N,t} =
|\psi_{N,t} \rangle \langle \psi_{N,t}|$ associated with
$\psi_{N,t}$ is defined as the orthogonal projection onto
$\psi_{N,t}$ (we use here Dirac's bracket notation; for $f,g,h \in
L^2 (\bR^d)$, $|f \rangle \langle g|: L^2 (\bR^d) \to L^2 (\bR^d)$
is the operator defined by $|f \rangle \langle g| (h) = \langle g,h
\rangle_{L^2} \; f$). The kernel of $\gamma_{N,t}$ is thus given by
\[ \gamma_{N,t} (\bx; \bx') = \psi_{N,t} (\bx) \overline{\psi}_{N,t}
(\bx') . \] For $k=1, \dots ,N$, we define then the $k$-particle
marginal density $\gamma^{(k)}_{N,t}$ associated with $\psi_{N,t}$
by taking the partial trace of $\gamma_{N,t}$ over the last $N-k$
particles. In other words, we define $\gamma^{(k)}_{N,t}$ as a
positive trace class operator on $L^2 (\bR^{3k})$ with kernel
\begin{equation}\label{eq:kpart}
\gamma_{N,t}^{(k)} (\bx_k ; \bx'_k) = \int \rd \bx_{N-k} \;
\gamma_{N,t} (\bx_k , \bx_{N-k} ; \bx'_k , \bx_{N-k})\,.
\end{equation}
Since $\| \psi_{N,t} \|_{L^2 (\bR^{3N})} = 1$, we immediately obtain
$\tr \; \gamma^{(k)}_{N,t} = 1$ for all $N\geq 1$, $k =1,\dots ,N$,
and $t \in \bR$.


\medskip

By the choice of the initial wave function (\ref{eq:init}), at time
$t=0$ we have $\gamma_{N,0}^{(k)} = |\ph \rangle \langle \ph
|^{\otimes k}$. It turns out that (\ref{eq:fac}) should be understood
in terms of convergence of marginal densities. For a large class
of interaction potentials $V$, for every fixed $k \geq 1$, and $t
\in \bR$, one can in fact show that
\begin{equation}\label{eq:conv} \gamma_{N,t}^{(k)} \to |\ph_t
\rangle \langle \ph_t|^{\otimes k} \qquad \text{as } N \to \infty
\end{equation} where $\ph_t$ is a solution of the nonlinear Hartree
equation (\ref{eq:hartree}). The convergence (\ref{eq:conv}) holds
in the trace norm topology. In particular, (\ref{eq:conv}) implies
that for arbitrary $k$ and for an arbitrary bounded operator
$J^{(k)}$ on $L^2 (\bR^{3k})$,
\[ \left\langle \psi_{N,t}, \left( J^{(k)} \otimes 1^{(N-k)} \right)
\psi_{N,t} \right\rangle \to \langle \ph_t^{\otimes k}, J^{(k)}
\ph_t^{\otimes k} \rangle \] as $N \to \infty$. The approximate identity
(\ref{eq:fac}) can thus be be interpreted as follows: as long as we are
interested in the expectation of observables depending non-trivially
only on a fixed number of particles, the $N$-body wave function
$\psi_{N,t}$ can be approximated by the $N$-fold tensor product of the
solution $\phi_t$ to the nonlinear Hartree equation (\ref{eq:hartree}).

\medskip

The first rigorous proof of (\ref{eq:conv}) was obtained by Spohn in
\cite{Spohn}, under the assumption of a bounded interaction
potential $V$. The problem of proving (\ref{eq:conv}) becomes
substantially more involved for singular potentials. In \cite{EY},
Erd\H os and Yau extended Spohn's approach to obtain a rigorous
derivation of the Hartree equation (\ref{eq:hartree}) for a Coulomb
interaction $V(x) = \const / |x|$ (partial results for the Coulomb
interaction were also obtained by Bardos, Golse, and Mauser in
\cite{BGM}). In \cite{ES}, the Hartree equation with Coulomb
interaction was derived for semirelativistic bosons; in the
semirelativistic setting, the dispersion of the bosons only grows
linearly in the momentum (for large momenta), and thus the control
of the Coulomb singularity is more delicate. In
\cite{EESY,ESY1,ESY2}, models described by the Hamiltonian
\[ H_N = \sum_{j=1}^N -\Delta_{x_j} + \frac{1}{N} \sum_{i<j}^N
N^{3\beta} V (N^{\beta} (x_i -x_j)) \qquad \text{with } \beta
\in(0,1] \] with an $N$-dependent potential were considered (in the
one-dimensional case, $N$-dependent potentials were considered by
Adami, Golse and Teta in \cite{AGT}). These models are used to
describe systems of physical interest, such as Bose-Einstein
condensates. Assuming the interaction to be positive ($V(x) \geq 0$
for all $x\in \bR^3$) and sufficiently small, the main result was
again a proof of the convergence (\ref{eq:conv}); this time,
however, $\ph_t$ is a solution of the cubic nonlinear Schr\"odinger
equation (with local nonlinearity)
\begin{equation}\label{eq:NLS} i\partial_t \ph_t = - \Delta \ph_t + \sigma
|\ph_t|^2 \ph_t \quad \text{with } \sigma
= \left\{ \begin{array}{ll} b_0 \quad &\text{if } 0< \beta < 1 \\
8\pi a_0 \quad &\text{if } \beta =1 \end{array} \right.  \; .
\end{equation} Here $b_0 = \int \rd x V(x)$ and $a_0$ is the scattering
length of $V$. The emergence of the scattering length $a_0$ for
$\beta=1$ (for all other choices of $0 < \beta <1$ the coupling
constant is given by $b_0$, which is the first Born approximation to
$8\pi a_0$) is a consequence of the short scale correlation
structure developed in solutions of the Schr\"odinger equation,
which, in the case $\beta =1$, is characterized by the same length
scale $O(1/N)$ as the scale of the interaction potential.

\medskip

The results described above have been obtained by extensions of the
approach introduced by Spohn in \cite{Spohn}, which was based on the
study of the BBGKY hierarchy
\begin{equation}\label{eq:BBGKY}
\begin{split}
i\partial_t \gamma^{(k)}_{N,t} = \; &\sum_{j=1}^k [-\Delta_{x_j},
\gamma^{(k)}_{N,t} ] + \frac{1}{N} \sum_{i<j}^k [V (x_i -x_j),
\gamma^{(k)}_{N,t}] \\ &+ \left( \frac{N-k}{N} \right) \sum_{j=1}^k
\tr_{k+1} \; [ V(x_j - x_{k+1}) , \gamma^{(k+1)}_{N,t} ]
\end{split}
\end{equation}
for the evolution of the marginal densities $\gamma^{(k)}_{N,t}$,
$k=1, \dots,N$ (here $\tr_{k+1}$ denotes the partial trace over the
$(k+1)$-th particle; this hierarchy is equivalent to the
Schr\"odinger equation (\ref{eq:schr}) for $\psi_{N,t}$). Because of
the compactness of the sequence $\gamma^{(k)}_{N,t}$, $N \geq k$,
the proof of (\ref{eq:conv}) reduces to two main steps. The first
step consists in proving that an arbitrary family of limit points
$\{ \gamma_{\infty,t}^{(k)} \}_{k \geq 1}$ satisfies the infinite
hierarchy \begin{equation}\label{eq:infi} i \partial_t
\gamma^{(k)}_{\infty,t} = \sum_{j=1}^k [-\Delta_{x_j},
\gamma^{(k)}_{\infty,t} ] + \sum_{j=1}^k \tr_{k+1} \; [ V(x_j -
x_{k+1}) , \gamma^{(k+1)}_{\infty,t} ] \, .\end{equation} The second
step is a proof of the uniqueness of the solution of
(\ref{eq:infi}). Since the factorized family
$\gamma_{\infty,t}^{(k)} = |\ph_t \rangle \langle \ph_t|^{\otimes
k}$, with $\ph_t$ determined by (\ref{eq:hartree}), is a solution of
the infinite hierarchy (\ref{eq:infi}), these two steps are
sufficient to obtain (\ref{eq:conv}).

\medskip

Despite its many successes, this method has some limitations. The
main one, from our point of view, is that, because of the use of
abstract arguments related to the compactness of the sequence
$\gamma_{N,t}^{(k)}$, this technique does not provide any
information on the rate of convergence of $\gamma_{N,t}^{(k)}$ to
$|\ph_t \rangle \langle \ph_t|^{\otimes k}$.

\medskip

In some cases, instead of comparing the solution of (\ref{eq:BBGKY}) with the
solution of the infinite hierarchy (\ref{eq:infi}), it is also
possible to expand it in a Duhamel series and to compare it directly
with the corresponding expansion for the factorized densities
$|\ph_t \rangle \langle \ph_t|^{\otimes k}$. This approach (see
\cite{Spohn}) leads to bounds of the form
\begin{equation}\label{eq:err1} \tr \; \left|\gamma_{N,t}^{(k)} -
|\ph_t \rangle \langle \ph_t|^{\otimes k} \right| \leq \frac{C^k}{N}
\end{equation} for all sufficiently small times $|t| \leq t_0$. The
restriction to small times is needed to guarantee the convergence of
the Duhamel expansion of the solution to (\ref{eq:BBGKY}). Iterating
the arguments used to obtain (\ref{eq:err1}), one can derive bounds
of the form
\[ \tr \; \left|\gamma_{N,t}^{(k)} - |\ph_t \rangle \langle
\ph_t|^{\otimes k} \right| \leq \frac{C^k}{N^{\frac{1}{2^t}}}\,
\] which hold for all $t \in \bR$, but deteriorate very fast in time
and are therefore not effective and not very useful. Next theorem,
which is the main result of this paper, provides much stronger
bounds on the difference between the true quantum mechanical
evolution of the marginal densities and their Hartree evolution; in
particular it shows that for every fixed time $t \in \bR$, the error
is at most of the order $O(N^{-1/2})$.

\begin{theorem}\label{thm:fact}
Suppose that there exists $D>0$ such that the operator inequality
\begin{equation}\label{eq:assump} V^2 (x) \leq D \; (1-\Delta_x)  \end{equation} holds true.
Let
\begin{equation}
\psi_N (\bx) = \prod_{j=1}^N \ph (x_j),\end{equation} for some $\ph
\in H^1 (\bR^3)$ with\footnote{In what follows, for a function $f$ we will always denote by
$\|f\|$ its $L^2$ norm, while, for an operator $A$, $\|A\|$ will mean its $L^2$ operator norm. }
$\| \ph \| =1$. Denote by $\psi_{N,t} = e^{-i
H_N t} \psi_N$ the solution to the Schr\"odinger equation
(\ref{eq:schr}) with initial data $\psi_{N,0}= \psi_N$, and let
$\gamma^{(1)}_{N,t}$ be the one-particle density associated with
$\psi_{N,t}$. Then there exist constants $C,K$, depending only on
the $H^1$ norm of $\ph$ and on the constant $D$ on the r.h.s. of
(\ref{eq:assump}) such that
\begin{equation}\label{eq:bd1}
\tr\; \Big| \gamma^{(1)}_{N,t} - |\ph_t \rangle \langle \ph_t| \Big|
\leq \frac{C}{N^{1/2}} \; e^{K t} \,.
\end{equation}
Here $\ph_t$ is the solution to the nonlinear Hartree equation
\begin{equation}\label{eq:hartree2}
i\partial_t \ph_t = -\Delta \ph_t + (V *|\ph_t|^2 ) \ph_t
\end{equation}
with initial data $\ph_{t=0} = \ph$.
\end{theorem}
\begin{remark}
The assumption on the potential $V$ means that the most singular potential
we can handle is the Coulomb potential $V(x) = \kappa / |x|$. Note
that our theorem applies both to the attractive ($\kappa <0$) and
the repulsive case ($\kappa >0$). In particular Theorem
\ref{thm:fact} implies the result obtained by Erd\"os and Yau in
\cite{EY}.
\end{remark}
\begin{remark}
Note that under the assumption (\ref{eq:assump}) on
the interaction potential $V$, the nonlinear equation
(\ref{eq:hartree2}) is known to be globally well-posed in $H^1
(\bR^3)$. This follows from the conservation of the mass $\| \ph \|$
and of the energy
\[ \cE (\ph) = \int \rd x \; |\nabla \ph (x)|^2 + \frac{1}{2} \int
\rd x \rd y \; V(x-y) |\ph (x)|^2 |\ph (y)|^2 \] and from the
observation that there exist constants $c_1, c_2$ such that
\begin{equation}\label{eq:H1ph}
\begin{split}
\cE (\ph) \leq c_1 \| \ph \|^2_{H^1} (1+ \| \ph\|^2) \quad \text{and
} \quad  \| \ph \|^2_{H^1} &\leq c_2 \left(\cE (\ph) + \|\ph\|^4+ \|\ph\|^2
\right).
\end{split}
\end{equation}
Both bounds can be proven using that, by (\ref{eq:assump}),
\[ \int \rd y \;  V(x-y) |\ph (y)|^2 \leq \eps \| \nabla \ph
\|^2 + \eps^{-1} \| \ph \|^2 \] for all $\eps >0$, uniformly in
$x\in \bR^3$.
\end{remark}

\begin{remark}
Instead of (\ref{eq:bd1}) we will prove that
\begin{equation}\label{eq:HS} \| \gamma^{(1)}_{N,t} - |\ph_t \rangle \langle \ph_t|
\|_{\text{HS}} \leq \frac{C}{N^{1/2}} \; e^{K t} \end{equation}
where $\| . \|_{\text{HS}}$ denotes the Hilbert-Schmidt norm.
Although in general the trace norm is bigger than the
Hilbert-Schmidt norm, in this case they differ at most by a factor
of two\footnote{We would like to thank Robert Seiringer for
pointing out this argument to us.}. In fact,
since $|\ph_t\rangle \langle \ph_t |$ is a rank one
projection, the operator $A= \gamma_{N,t}^{(1)} - |\ph_t\rangle
\langle \ph_t|$ can only have one negative eigenvalue
$\lambda_{\text{neg}} <0$. Since moreover \[ \tr \;
\left(\gamma_{N,t}^{(1)} -|\ph_t\rangle \langle \ph_t|\right) = 0 \]
it follows that the negative eigenvalue of $A$ is equal, in absolute
value, to the sum of all positive eigenvalues. The trace norm of $A$
is equal, therefore, to $2|\lambda_{\text{neg}}| = 2 \| A \|$, where
$\| A \|$ denotes the operator norm of $A$. Since $\| A \| \leq \| A
\|_{\text{HS}}$, we immediately obtain that $\tr \; | A | \leq 2 \|
A \|_{\text{HS}}$.
\end{remark}

\begin{remark}
The bound (\ref{eq:bd1}) is not optimal. As mentioned
above, for short times and bounded potentials, the quantity on the
l.h.s. of (\ref{eq:bd1}) is known to be of the order $1/N$.
Nevertheless Theorem~\ref{thm:fact} is the first estimate on the
rate of convergence towards the mean-field limit which holds for all
times and remains of the same order $N^{-1/2}$ for all fixed times.
\end{remark}
\begin{remark}
Although, in order to simplify the analysis, we
only consider the rate of convergence of the one-particle density
$\gamma_{N,t}^{(1)}$ to $|\ph_t \rangle \langle \ph_t|$, our method
can also be used to prove bounds of the form
\[ \tr \; \left| \gamma^{(j)}_{N,t} -|\ph_t \rangle \langle
\ph_t|^{\otimes j} \right| \leq \frac{C (j)}{N^{1/2}} \; e^{K(j) \, t} \]
for all $j,t,N$ and for $j$-dependent constants $C(j), K(j)$.
\end{remark}

In this paper we avoid the use of the BBGKY hierarchy and instead revive
an approach, introduced by Hepp in \cite{Hepp} and
extended by Ginibre and Velo in \cite{GV}, to the study of a semiclassical limit of
quantum many-boson
systems\footnote{Mathematically, the
semiclassical limit considered in \cite{Hepp,GV} is equivalent to
the mean field limit considered in the present manuscript.}.  This
approach is based on embedding the $N$-body Schr\"odinger system
into the second quantized Fock-space representation and on the use of
coherent states as initial data. The use of the Fock-space representation
is in particular dictated by the fact
that coherent states  do not have a fixed number of particles.

The Hartree dynamics emerges as the main component of the evolution of coherent states
in the mean field limit (or, in the language of \cite{Hepp,GV}, in the semiclassical
limit). The problem then reduces to the study of quantum fluctuations, described by an
$N$-dependent two-parameter unitary group $\cU_N (t;s)$,
around the Hartree dynamics. In \cite{Hepp,GV}, Hepp (for smooth interaction
potentials) and Ginibre and Velo (for singular potentials) proved
that, in the limit $N \to \infty$, the fluctuation dynamics $\cU_N
(t;s)$ approaches a limiting evolution $\cU (t;s)$. This important
result shows the relevance of the Hartree dynamics in the mean field
limit (at least in the case of coherent initial states). It does not
prove, however, the convergence (\ref{eq:conv}) of the one-particle
marginal density to the orthogonal projection onto the solution of
the Hartree equation, nor does it imply convergence results for the evolution of factorized
initial sates. The problem of convergence of marginals requires additional
control on the growth of the number\footnote{Fluctuations around the Hartree dynamics
will be considered as particle excitations and thus it will be
possible to compute their number.} of fluctuations generated by the
evolution $\cU_N (t;s)$. This analysis, which,
technically, is the most difficult part of the present paper (see
Proposition \ref{prop:1cou}), is new\footnote{A more precise
discussion of the results of \cite{Hepp,GV}, and of their relation
with our work can be found at the end of Section \ref{sec:coh}.}.
Another novel part of our work is the derivation of convergence towards Hartree
dynamics for factorized initial sates from the corresponding statements for the
evolution of coherent states.

\medskip

Although we are mainly concerned with the dynamics of factorized
initial data, the result we obtain for coherent states (see Theorem
\ref{thm:coh}) is of independent interest, especially because, in
this case, our bound is optimal in its $N$-dependence (for coherent
states, we show that the error is at most of the order $1/N$ for
every fixed time).

\medskip

The paper is organized as follows. In Section \ref{sec:Fock}, we
define the Fock space representation of the mean field system,
introduce coherent states and review their main properties. In
Section \ref{sec:coh}, we consider the evolution of a coherent state
and we prove that, in this case, the rate of convergence to the mean
field solution remains of the order $1/N$ for all fixed times.
Finally, in Section \ref{sec:fact}, we show how to use coherent
states to obtain information on the dynamics of factorized states,
and we prove Theorem \ref{thm:fact}.

\section{Fock space representation}\label{sec:Fock}
\setcounter{equation}{0}

We define the bosonic Fock space over $L^2 (\bR^3, \rd x)$ as the
Hilbert space
\[ \cF = \bigoplus_{n \geq 0} L^2 (\bR^3 , \rd x)^{\otimes_s n} =
\bC \oplus \bigoplus_{n \geq 1} L^2_s (\bR^{3n}, \rd x_1 \dots \rd
x_n)\, ,
\] with the convention $L^2 (\bR^3)^{\otimes_s 0} = \bC$.
Vectors in $\cF$ are sequences $\psi = \{ \psi^{(n)} \}_{n \geq 0}$
of $n$-particle wave functions $\psi^{(n)} \in L^2_s (\bR^{3n})$.
The scalar product on $\cF$ is defined by
\[ \langle \psi_1 , \psi_2 \rangle = \sum_{n \geq 0} \langle
\psi_1^{(n)} , \psi_2^{(n)} \rangle_{L^2 (\bR^{3n})} =
\overline{\psi_1^{(0)}} \psi_2^{(0)} + \sum_{n \geq 1} \int \rd x_1
\dots \rd x_n \, \overline{\psi_1^{(n)}} (x_1 , \dots , x_n)
\psi_2^{(n)} (x_1, \dots ,x_n) \,. \] An $N$ particle state with
wave function $\psi_N$ is described on $\cF$ by the sequence $\{
\psi^{(n)} \}_{ n \geq 0}$ where $\psi^{(n)} = 0$ for all $n \neq N$
and $\psi^{(N)} = \psi_N$. The vector $\{1, 0, 0, \dots \} \in \cF$
is called the vacuum, and will be denoted by $\Omega$.

On $\cF$, we define the number of particles operator $\cN$, by $(\cN
\psi)^{(n)} = n \psi^{(n)}$. Eigenvectors of $\cN$ are vectors of
the form $\{ 0, \dots, 0, \psi^{(m)}, 0,  \dots \}$ with a fixed
number of particles. For $f \in L^2 (\bR^3)$ we also define the
creation operator $a^* (f)$ and the annihilation operator $a(f)$ on
$\cF$ by
\begin{equation}
\begin{split}
\left(a^* (f) \psi \right)^{(n)} (x_1 , \dots ,x_n) &=
\frac{1}{\sqrt n} \sum_{j=1}^n f(x_j) \psi^{(n-1)} ( x_1, \dots,
x_{j-1}, x_{j+1},
\dots , x_n) \\
\left(a (f) \psi \right)^{(n)} (x_1 , \dots ,x_n) &= \sqrt{n+1} \int
\rd x \; \overline{f (x)} \, \psi^{(n+1)} (x, x_1, \dots ,x_n) \, .
\end{split}
\end{equation}
The operators $a^* (f)$ and $a(f)$ are unbounded, densely defined,
closed operators. The creation operator $a^*(f)$ is the adjoint of
the annihilation operator $a(f)$ (note that by definition $a(f)$ is
anti-linear in $f$), and they satisfy the canonical commutation
relations \begin{equation}\label{eq:comm} [ a(f) , a^* (g) ] =
\langle f,g \rangle_{L^2 (\bR^3)}, \qquad [ a(f) , a(g)] = [ a^*
(f), a^* (g) ] = 0 \,. \end{equation} For every $f\in L^2 (\bR^3)$,
we introduce the self adjoint operator
\[ \phi (f) = a^* (f) + a(f) \,. \]

We will also make use of operator valued distributions $a^*_x$ and
$a_x$ ($x \in \bR^3$), defined so that \begin{equation}\begin{split}
a^* (f) &= \int \rd x \, f(x) \, a_x^* \\ a(f) & = \int \rd x \,
\overline{f (x)} \, a_x \end{split}
\end{equation}
for every $f \in L^2 (\bR^3)$. The canonical commutation relations
assume the form \[ [ a_x , a^*_y ] = \delta (x-y) \qquad [ a_x, a_y
] = [ a^*_x , a^*_y] = 0 \, .\]

The number of particle operator, expressed through the distributions
$a_x,a^*_x$, is given by
\[ \cN = \int \rd x \, a_x^* a_x \,. \]

The following lemma provides some useful bounds to control creation
and annihilation operators in terms of the number of particle
operator $\cN$.
\begin{lemma}\label{lm:a-bd}
Let $f \in L^2 (\bR^3)$. Then
\begin{equation}
\begin{split}
\| a(f) \psi \| &\leq \| f \| \, \| \cN^{1/2} \psi \| \\
\| a^* (f) \psi \| &\leq \| f \| \, \| \left( \cN + 1 \right)^{1/2}
\psi \| \\
\| \phi (f) \psi \| &\leq 2 \| f \| \| \left( \cN + 1 \right)^{1/2}
\psi \|\,
\end{split}
\end{equation}
\end{lemma}
\begin{proof}
The last inequality clearly follows from the first two. To prove the
first bound we note that
\begin{equation}
\begin{split} \| a (f) \psi \| &\leq \int \rd x \, |f(x)| \, \| a_x
\psi \| \leq \left( \int \rd x \, |f(x)|^2 \right)^{1/2} \, \left(
\int \rd x \, \| a_x \psi \|^2 \right)^{1/2} \\ &= \| f \| \, \|
\cN^{1/2} \psi \|\, .
\end{split}
\end{equation}
The second estimate follows by the canonical commutation relations
(\ref{eq:comm}) because
\begin{equation}
\begin{split}
\| a^* (f) \psi \|^2 &= \langle \psi, a(f) a^* (f) \psi \rangle =
\langle \psi, a^* (f) a(f) \psi \rangle + \| f \|^2 \| \psi \|^2  \\
&= \| a(f) \psi \|^2 + \| f\|^2 \| \psi \|^2 \leq \| f\|^2 \, \left(
\|\cN^{1/2} \psi \| + \| \psi \|^2 \right) = \| f \|^2 \| \left( \cN
+1 \right)^{1/2} \psi \|^2 \, .
\end{split}
\end{equation}
\end{proof}

Given $\psi \in \cF$, we define the one-particle density
$\gamma^{(1)}_{\psi}$ associated with $\psi$ as the positive trace
class operator on $L^2 (\bR^3)$ with kernel given by
\begin{equation}\label{eq:margi} \gamma^{(1)}_{\psi} (x; y) = \frac{1}{\langle \psi,
\cN \psi \rangle} \, \langle \psi, a_y^* a_x \psi \rangle\, .
\end{equation} By definition, $\gamma_{\psi}^{(1)}$ is a positive trace
class operator on $L^2 (\bR^3)$ with $\tr \, \gamma_{\psi}^{(1)}
=1$. For every $N$-particle state with wave function $\psi_N \in
L^2_s (\bR^{3N})$ (described on $\cF$ by the sequence $\{ 0, 0,
\dots, \psi_N, 0,0, \dots \}$) it is simple to see that this
definition is equivalent to the definition (\ref{eq:kpart}).

\medskip

We define the Hamiltonian $\cH_N$ on $\cF$ by $ (\cH_N \psi)^{(n)} =
\cH^{(n)}_N \psi^{(n)}$, with
\[ \cH^{(n)}_N = - \sum_{j=1}^n \Delta_j + \frac{1}{N} \sum_{i<j}^n
V(x_i -x_j) \, . \] Using the distributions $a_x, a^*_x$, $\cH_N$
can be rewritten as
\begin{equation}\label{eq:ham2} \cH_N = \int \rd x \nabla_x a^*_x
\nabla_x a_x + \frac{1}{2N} \int \rd x \rd y \, V(x-y) a_x^* a_y^*
a_y a_x \, . \end{equation} By definition the Hamiltonian $\cH_N$
leaves sectors of $\cF$ with a fixed number of particles invariant.
Moreover, it is clear that on the $N$-particle sector, $\cH_N$
agrees with the Hamiltonian $H_N$ (the subscript $N$ in $\cH_N$ is a
reminder of the scaling factor $1/N$ in front of the potential
energy). We will study the dynamics generated by the operator
$\cH_N$. In particular we will consider the time evolution of
coherent states, which we introduce next.

\medskip

For $f \in L^2 (\bR^3)$, we define the Weyl-operator
\begin{equation}
W(f) = \exp \left( a^* (f) - a(f) \right) = \exp \left( \int \rd x
\, (f(x) a^*_x - \overline{f} (x) a_x) \right) \, .
\end{equation}
Then the coherent state $\psi (f) \in \cF$ with one-particle wave
function $f$ is defined by
\[ \psi (f) = W(f) \Omega \, .\]
Notice that \begin{equation}\label{eq:coh} \psi (f)= W(f) \Omega =
e^{-\| f\|^2 /2} \sum_{n \geq 0} \frac{ (a^* (f))^n}{n!} \Omega  =
e^{-\| f\|^2 /2} \sum_{n \geq 0} \frac{1}{\sqrt{n!}} \, f^{\otimes n} \,,
\end{equation}
where $f^{\otimes n}$ indicates the Fock-vector $\{ 0, \dots , 0 ,f^{\otimes n}, 0, \dots \}$. This follows from
\[ \exp (a^* (f) - a (f)) = e^{-\|f \|^2/2} \exp (a^* (f)) \exp
(-a(f)) \] which is a consequence of the fact that the commutator $[
a (f) , a^* (f)] = \| f \|^2$ commutes with $a(f)$ and $a^* (f)$.
{F}rom Eq. (\ref{eq:coh}) we see that coherent states are
superpositions of states with different number of particles (the
probability of having $n$ particles in $\psi (f)$ is given by
$e^{-\| f\|^2} \| f \|^{2n}/n!$).

\medskip

In the following lemma we collect some important and well known
properties of Weyl operators and coherent states.
\begin{lemma}\label{lm:coh}
Let $f,g \in L^2 (\bR^3)$.
\begin{itemize}
\item[i)] The Weyl operator satisfy the relations
\[ W(f) W(g) = W(g) W(f) e^{-2i \, \text{Im} \, \langle f,g \rangle} = W(f+g) e^{-i\, \text{Im} \, \langle f,g \rangle} \,. \]
\item[ii)] $W(f)$ is a unitary operator and
\[ W(f)^* = W(f)^{-1}  = W (-f). \]
\item[iii)] We have \[ W^* (f) a_x W(f) = a_x + f(x), \qquad \text{and} \quad W^* (f) a^*_x
W(f) = a^*_x + \overline{f} (x) \, .\]
\item[iv)] {F}rom iii) we see that coherent states are eigenvectors of annihilation operators
\[ a_x \psi (f) = f(x) \psi (f)  \qquad \Rightarrow \qquad a (g)
\psi (f) = \langle g, f \rangle_{L^2} \psi (f) \, .\]
\item[v)] The expectation of the number of particles in the coherent
state $\psi (f)$ is given by $\| f\|^2$, that is
\[ \langle \psi (f), \cN \psi (f) \rangle = \| f \|^2
\, . \] Also the variance of the number of particles in $\psi (f)$
is given by $\|f \|^2$ (the distribution of $\cN$ is Poisson), that
is
\[ \langle \psi (f), \cN^2 \psi (f) \rangle - \langle \psi (f) ,
\cN \psi (f) \rangle^2 = \| f \|^2 \, .\]
\item[vi)] Coherent states are normalized but not orthogonal to each
other. In fact
\[ \langle \psi (f) , \psi (g) \rangle = e^{-\frac{1}{2}\left( \| f
\|^2 + \| g \|^2 - 2 (f,g) \right)}  \quad \Rightarrow \quad
|\langle \psi (f) , \psi (g) \rangle| = e^{-\frac{1}{2} \| f- g
\|^2} \, .\]
\end{itemize}
\end{lemma}

\section{Time evolution of coherent states}\label{sec:coh}
\setcounter{equation}{0}

Next we study the dynamics of coherent states with expected number
of particles $N$ in the limit $N \to \infty$. We choose the initial
data
\begin{equation} \psi (\sqrt{N} \ph) = W(\sqrt{N} \ph) \Omega \qquad
\text{for $\ph \in H^1 (\bR^3)$ with } \| \ph \| =1 \end{equation}
and we study its time evolution $\psi (N,t) = e^{-i \cH_N t} \psi
(\sqrt{N} \ph)$ with the Hamiltonian $\cH_N$ defined in
(\ref{eq:ham2}).

\begin{theorem}\label{thm:coh}
Suppose that there exists $D >0$ such that the operator inequality
\begin{equation}\label{eq:assump-coh} V^2 (x) \leq D (1 - \Delta_x) \end{equation}
holds true. Let $\Gamma_{N,t}^{(1)}$ be the one-particle marginal
associated with $\psi(N,t)= e^{-i\cH_N t} W(\sqrt{N} \ph) \Omega$ (as defined
in (\ref{eq:margi})). Then there exist constants $C,K>0$ (only
depending on the $H^1$-norm of $\ph$ and on the constant $D$
appearing in (\ref{eq:assump-coh})) such that
\begin{equation}\label{eq:tr-1}
\tr \; \Big|\Gamma^{(1)}_{N,t} - |\ph_t \rangle \langle \ph_t| \Big|
\leq \frac{C}{N} \; e^{K t}
\end{equation}
for all $t \in \bR$.
\end{theorem}
\begin{remark}
The use of coherent states as initial data allows us
to obtain the optimal rate of convergence $1/N$ for all fixed times
(while for the evolution of factorized $N$-particle states we only
get the rate $1/\sqrt{N}$; see (\ref{eq:bd1})).
\end{remark}
\begin{proof}
The proof of Theorem \ref{thm:coh} will occupy the remaining subsections
of section \ref{sec:coh}.
\subsection{Dynamics $\cU_N$ of quantum fluctuations}
By (\ref{eq:margi}), the kernel of $\Gamma^{(1)}_{N,t}$ is given by
\begin{equation}\label{eq:gamma}
\begin{split}
\Gamma^{(1)}_{N,t} (x;y) = &\frac{1}{N} \left\langle \Omega, W^*
(\sqrt{N} \ph) e^{i\cH_N t} a_y^* a_x e^{-i\cH_N t} W(\sqrt{N} \ph)
\Omega \right\rangle \\ = \; &\ph_t (x) \overline{\ph}_t (y) +
\frac{\overline{\ph}_t (y)}{\sqrt{N}} \left\langle \Omega, W^*
(\sqrt{N} \ph) e^{i\cH_N t} (a_x - \sqrt{N} \ph_t(x)) e^{-i\cH_N t}
W(\sqrt{N} \ph) \Omega \right\rangle \\ &+ \frac{\ph_t
(x)}{\sqrt{N}} \left\langle \Omega, W^* (\sqrt{N} \ph) e^{i\cH_N t}
(a_y^* - \sqrt{N} \overline{\ph}_t (y)) e^{-i\cH_N t} W(\sqrt{N}
\ph) \Omega \right\rangle \\ &+ \frac{1}{N} \left\langle \Omega, W^*
(\sqrt{N} \ph) e^{i\cH_N t} ( a_y^* - \sqrt{N} \overline{\ph}_t (y))
(a_x - \sqrt{N} \ph_t (x)) e^{-i \cH_N t} W(\sqrt{N} \ph ) \Omega
\right\rangle \,.
\end{split}
\end{equation}
It was observed by Hepp in \cite{Hepp} (see also
Eqs. (1.17)-(1.28) in \cite{GV}) that
\begin{equation}\label{eq:GV}
\begin{split} W^* (\sqrt{N} \ph_s) \; e^{i\cH_N (t-s)} (a_x - \sqrt{N} \ph_t
(x)) e^{-i\cH_N (t-s)} W (\sqrt{N} \ph_s) &= \cU_N  (t;s)^* \, a_x \,
\cU_N (t;s) \\ &= \cU_N (s;t) \, a_x \, \cU_N (t;s)
\end{split}\end{equation} where the unitary evolution $\cU_N (t;s)$
is determined by the equation\footnote{Note that, explicitly,
$\cU_N(t,s)=W^*(\sqrt N \phi_t) e^{-i\cH_N(t-s)} W(\sqrt N\phi_s)$.}
\begin{equation}\label{eq:cUN} i
\partial_t \cU_N (t;s) = \cL_N (t) \cU_N (t;s) \qquad \text{and}
\quad \cU_N (s;s)= 1 \end{equation} with the generator
\begin{equation}\label{eq:cLN}
\begin{split}
\cL_N (t) = & \int \rd x \, \nabla_x a^*_x \nabla_x a_x + \int \rd x
\, \left(V *|\ph_t|^2 \right) (x) \, a^*_x a_x +
\int \rd x \rd y \, V(x-y) \, \overline{\ph_t} (x) \ph_t (y) a^*_y a_x \\
&+ \frac{1}{2} \int \rd x \rd y \, V(x-y) \left( \ph_t (x) \ph_t (y)
a^*_x a^*_y +
\overline{\ph_t} (x) \overline{\ph_t} (y) a_x a_y \right) \\
&+\frac{1}{\sqrt{N}} \int \rd x \rd y \, V(x-y) \, a_x^* \left(
\ph_t (y)
a^*_y  + \overline{\ph_t} (y) a_y \right) a_x \\
&+\frac{1}{2N} \int \rd x \rd y \, V(x-y) \, a^*_x a^*_y a_y a_x \, .
\end{split}
\end{equation}
It follows from (\ref{eq:gamma}) that
\begin{equation}\label{eq:gamma2cou}
\begin{split}
\Gamma^{(1)}_{N,t} (x,y) - \ph_t (x) \overline{\ph}_t (y) = \;
&\frac{1}{N} \left\langle \Omega, \cU_N (t;0)^* a_y^* a_x \cU_N
(t;0) \Omega \right\rangle \\ &+ \frac{\ph_t (x)}{\sqrt{N}}
\left\langle \Omega,\cU_N (t;0)^* a^*_y \cU_N (t;0) \Omega \right\rangle \\
&+ \frac{\overline{\ph}_t (y)}{\sqrt{N}} \left\langle \Omega,\cU_N
(t;0)^* a_x \cU_N (t;0) \Omega \right\rangle\,.
\end{split}
\end{equation}
In order to produce another decaying factor $1/\sqrt{N}$ in the last
two term on the r.h.s. of the last equation, we compare the
evolution $\cU_N (t;0)$ with another evolution $\wt\cU_N (t;0)$
defined through the equation
\begin{equation}\label{eq:wtcUN} i
\partial_t \wt \cU_N (t;s) = \wt\cL_N (t)\, \wt\cU_N (t;s) \qquad
\text{with} \quad \wt\cU_N (s;s) = 1 \end{equation} with the
time-dependent generator
\begin{equation}
\begin{split}
\wt \cL_N (t) = & \int \rd x \, \nabla_x a^*_x \nabla_x a_x + \int
\rd x \, \left(V *|\ph_t|^2 \right) (x) \, a^*_x a_x + \int \rd x
\rd y \, V(x-y) \overline{\ph_t} (x) \ph_t (y) a^*_y a_x \\
&+ \frac{1}{2} \int \rd x \rd y \, V(x-y) \left( \ph_t (x) \ph_t (y)
a^*_x a^*_y +
\overline{\ph_t} (x) \overline{\ph_t} (y) a_x a_y \right) \\
&+\frac{1}{2N} \int \rd x \rd y \, V(x-y) \, a^*_x a^*_y a_y a_x \, .
\end{split}
\end{equation}
{F}rom (\ref{eq:gamma2cou}) we find
\begin{equation}\label{eq:gamma3cou}
\begin{split}
&\Gamma^{(1)}_{N,t} (x;y) - \ph_t (x) \overline{\ph}_t (y) \\ & =
\frac{1}{N} \langle \Omega, \cU_N (t;0)^* a_y^* a_x \cU_N (t;0)
\Omega \rangle \\ &\;\; + \frac{\ph_t (x)}{\sqrt{N}} \left(
\left\langle \Omega,\cU_N (t;0)^* a_y^* \left(\cU_N (t;0) - \wt
\cU_N (t;0)\right) \Omega \right\rangle + \left\langle \Omega,
\left(\cU_N (t;0)^* - \wt \cU_N (t;0)^* \right) a_y^* \wt \cU_N
(t;0) \Omega
\right\rangle \right)\\
&\;\; + \frac{\overline{\ph}_t (y)}{\sqrt{N}} \left( \left\langle
\Omega,\cU_N (t;0)^* a_x \left( \cU_N (t;0) - \wt \cU_N (t;0)
\right) \Omega \right\rangle + \left\langle \Omega, \left(\cU_N
(t;0)^* -\wt \cU_N (t;0)^* \right) a_x \wt \cU_N (t;0) \Omega
\right\rangle \right).
\end{split}
\end{equation}
Here we used the fact that
\[ \left\langle \Omega, \wt \cU_N (t;0)^*  a_y \, \wt \cU_N (t;0) \Omega
\right\rangle = \left\langle \Omega,\wt \cU_N (t;0)^* a_x^*\,  \wt
\cU_N (t;0) \Omega \right\rangle = 0 \, .\] This follows from the
observation that, although the evolution $\wt\cU_N (t)$ does not
preserve the number of particles, it preserves the parity (it
commutes with $(-1)^{\cN}$). Multiplying (\ref{eq:gamma3cou}) with
the kernel $J(x,y)$ of a Hilbert-Schmidt operator $J$ over $L^2
(\bR^3)$ and taking the trace, we obtain
\begin{equation*}
\begin{split}
\tr \, J &\left( \Gamma^{(1)}_{N,t} - |\ph_t \rangle \langle \ph_t|
\right) \\ = \; &\frac{1}{N} \int \rd x \rd y \, J (x,y)
\langle a_y \cU_N (t;0) \Omega , a_x \cU_N (t;0) \Omega \rangle \\
&+ \frac{1}{\sqrt{N}} \int \rd x \rd y \, J (x,y) \ph_t (x) \langle
a_y \cU_N (t;0) \Omega, \left(\cU_N
(t;0) - \wt \cU_N (t;0)\right) \Omega \rangle \\
&+ \frac{1}{\sqrt{N}} \int \rd x \rd y \, J (x,y) \ph_t (x) \langle
\left(\cU_N (t;0) - \wt \cU_N (t;0) \right)\Omega,  a^*_y \wt
\cU_N (t;0) \Omega \rangle \\
&+ \frac{1}{\sqrt{N}} \int \rd x \rd y \, J (x,y) \,
\overline{\ph}_t (y) \langle  a_x^* \cU_N (t;0) \Omega, \left( \cU_N
(t;0) - \wt \cU_N (t;0) \right) \Omega \rangle \\ & +
\frac{1}{\sqrt{N}} \int \rd x \rd y J (x,y) \, \overline{\ph}_t (y)
\langle \left(\cU_N (t;0) -\wt \cU_N (t;0) \right)\Omega,  a_x \wt
\cU_N (t;0) \Omega \rangle\,.
\end{split}
\end{equation*}
Hence
\begin{equation*}
\begin{split}
\Big| \tr \, J &\left( \Gamma^{(1)}_{N,t} - |\ph_t \rangle \langle
\ph_t| \right) \Big| \\ \leq \; &\frac{1}{N} \left(\int \rd
x \rd y \, |J (x,y)|^2 \right)^{1/2} \; \int \rd x \| a_x \cU_N (t;0) \Omega \|^2 \,\\
&+ \frac{1}{\sqrt{N}} \left\| \left(\cU_N (t;0) - \wt \cU_N (t;0)
\right) \Omega \right\|
\, \int \rd x \, |\ph_t (x)| \| a (J (x,.)) \cU_N (t;0) \Omega \| \\
&+ \frac{1}{\sqrt{N}} \left\|\left(\cU_N (t;0) - \wt \cU_N (t;0)
\right)\Omega\right\| \, \int \rd x |\ph_t (x)| \| a^* (J (x,.)) \wt
\cU_N (t;0) \Omega \| \\
&+ \frac{1}{\sqrt{N}} \left\|  \left( \cU_N (t;0) - \wt \cU_N (t;0)
\right) \Omega \right\|  \int \rd y \, |\ph_t (y)| \| a^* (J (.,y))
\cU_N (t;0) \Omega \| \\ & + \frac{1}{\sqrt{N}} \left\|\left(\cU_N
(t;0) -\wt \cU_N (t;0) \right)\Omega \right\| \, \int \rd y  \,
|\ph_t (y)| \| a (J(.,y)) \wt \cU_N (t;0) \Omega \|
\end{split}
\end{equation*}
and therefore
\begin{equation*}
\begin{split}
\Big| \tr \, J \left( \Gamma^{(1)}_{N,t} - |\ph_t \rangle
\langle \ph_t| \right) \Big|
\leq \; &\frac{\| J \|_{\text{HS}}}{N} \; \langle \cU_N (t;0) \Omega, \cN \cU_N (t;0) \Omega \rangle \\
&+ \frac{2 \| J \|_{\text{HS}}}{\sqrt{N}} \| (\cU_N (t;0) - \wt
\cU_N (t;0)) \Omega \|
\, \| (\cN+1)^{1/2} \cU_N (t;0) \Omega \| \\
&+ \frac{2 \| J \|_{\text{HS}}}{\sqrt{N}} \|(\cU_N (t;0) - \wt \cU_N
(t;0))\Omega\| \, \| (\cN+1)^{1/2} \wt \cU_N (t;0) \Omega \|\,.
\end{split}
\end{equation*}
The proof of Theorem  \ref{thm:coh}
now follows from Proposition \ref{prop:1cou}, Lemma
\ref{lm:2cou}, Lemma \ref{lm:3cou}, and from the remark that the
trace norm can be controlled, in this case, by twice the
Hilbert-Schmidt norm (see Remark 3 after Theorem \ref{thm:fact}).
\end{proof}

\begin{proposition}\label{prop:1cou}
Let $\cU_N (t;s)$ be the unitary evolution defined in
(\ref{eq:cUN}). Then there exists a constant $K$, and, for every $j \in \bN$, constants $C(j), K(j)$ (depending only on $\| \ph \|_{H^1}$ and on the constant $D$ appearing in
(\ref{eq:assump-coh})) such that
\begin{equation}\label{eq:lm1cou} \langle \cU_N (t;s) \psi, \cN^j \cU_N
(t;s) \psi \rangle \leq C(j) \langle \psi, (\cN+1)^{2j+2} \psi \rangle \, e^{K(j) |t-s|} \,. \end{equation}
for all $\psi \in \cF$, and for all $t,s \in \bR$.
\end{proposition}
\begin{remark}
Proposition \ref{prop:1cou} states that the number of particles produced by the
dynamics $\cU_N$ of quantum fluctuations is independent of $N$ and grows in
time with at most exponential rate. This $N$-independence plays an important role
in our analysis. Its proof requires the introduction of yet another dynamics $\cU_N^{(M)}$,
whose generator looks very similar to $\cL_N(t)$ but contains a cutoff, in the cubic term, guaranteeing that the number of particles is smaller than a given $M$.
\end{remark}
\begin{proof}
We start by introducing a new unitary dynamics with time-dependent
generator $\cL^{(M)}_N (t)$ similar to $\cL_N (t)$ but with a cutoff
in the number of particles in the cubic term.

\subsection{Truncated dynamics $\cU_N^{(M)}$}

For a fixed $M >0$ (at the end we will choose $M=N$), we consider the
time-dependent generator
\begin{equation}
\begin{split}
\cL^{(M)}_N (t) = &\int \rd x \, \nabla_x a^*_x \nabla_x a_x + \int
\rd x \, \left(V *|\ph_t|^2 \right) (x) \, a^*_x a_x + \int \rd x
\rd y \, V(x-y) \, \overline{\ph_t} (x) \ph_t (y) a^*_y a_x \\
&+ \frac{1}{2} \int \rd x \rd y \, V(x-y) \left( \ph_t (x) \ph_t (y)
a^*_x a^*_y +
\overline{\ph_t} (x) \overline{\ph_t} (y) a_x a_y \right) \\
&+\frac{1}{\sqrt{N}} \int \rd x \rd y \, V(x-y) \, a^*_x \left(
\overline{\ph}_t (y) a_y \chi (\cN \leq M) + \ph_t (y) \chi
(\cN \leq M) a_y^* \right) a_x \\
&+\frac{1}{2N} \int \rd x \rd y \, V(x-y) \, a^*_x a^*_y a_y a_x \,
\end{split}
\end{equation}
and the corresponding time-evolution $\cU_N^{(M)} (t;s)$, defined by
\[ i\partial_t \cU_N^{(M)} (t;s) = \cL^{(M)}_N (t) \cU_N^{(M)} (t;s)
\qquad \text{with} \qquad \cU_N^{(M)} (s;s) = 1 \, .\]

\bigskip

\noindent
{\it Step 1. in the proof of Proposition \ref{prop:1cou}}
\begin{lemma}\label{lem:trunc}
There exists a constant $K$ (only depending on $\| \ph
\|_{H^1}$ and on the constant $D$ in (\ref{eq:assump-coh})), such
that, for all $N,M \in \bN$, $\psi \in \cF$, and $t,s \in \bR$
\begin{equation}\label{eq:step1cou} \langle \cU_N^{(M)} (t;s)
\psi, \cN^j \cU_N^{(M)} (t;s) \psi \rangle \leq   \langle \psi, (\cN+1)^j \psi
\rangle \; \exp \left(4^{j} \, K  |t-s| (1+ \sqrt{M/N})\right)\,.
\end{equation}
\end{lemma}
\begin{proof}[Proof of Lemma \ref{lem:trunc}.]
To prove
(\ref{eq:step1cou}) we compute the time-derivative of the
expectation of $(\cN+1)^j$. It suffices to consider the case $s=0$. We find
\begin{align*}
\frac{\rd}{\rd t} \langle \cU_N^{(M)} &(t;0) \psi, (\cN+1)^j
\cU_N^{(M)} (t;0) \psi \rangle \\ = \; &\langle \cU_N^{(M)} (t;0)
\psi, [i\cL^{(M)}_N (t), (\cN+1)^j ] \cU_N^{(M)} (t;0) \psi
\rangle
\\ = \; & \text{Im} \, \int \rd x \rd y V(x-y) \, \ph_t (x) \ph_t (y) \langle \cU_N^{(M)} (t;0)
\psi, [a_x^* a_y^* , (\cN+1)^j ] \cU_N^{(M)} (t;0) \psi \rangle \\
+&\frac 2{\sqrt N} \, \text{Im} \, \int \rd x \rd y V(x-y) \, \overline{\ph}_t (y)
\langle \cU_N^{(M)} (t;0)
\psi, [a_x^* a_y \chi (\cN \leq M) a_x, (\cN+1)^j]\, \cU_N^{(M)} (t;0) \psi \rangle
\end{align*}
Using the pull-through formulae $a_x \cN = (\cN+1) a_x$, $a_x^* \cN
= (\cN-1) a_x^*$, we find
$$
[a_x^*, (\cN+1)^j ]=\sum_{k=0}^{j-1} {j \choose k} (-1)^k (\cN+1)^k a_x^*,\qquad
[a_x, (\cN+1)^j ]=\sum_{k=0}^{j-1} {j \choose k}  (\cN+1)^k a_x.
$$
As a consequence,
\begin{align*}
[a_x^*a_y^*, (\cN+1)^j ]&=\sum_{k=0}^{j-1} {j \choose k} (-1)^k \left (a_x^* (\cN+1)^k a_y^*+ (\cN+1)^k a_x^*a_y^*\right)\\ &=\sum_{k=0}^{j-1} {j \choose k} (-1)^k \left ( \cN^{\frac k2} a_x^* a_y^* (\cN+2)^{\frac k2} +
(\cN+1)^{\frac k2} a_x^*a_y^* (\cN+3)^{\frac k2}\right),\\
[a_x, (\cN+1)^j ]&=\sum_{k=0}^{j-1} {j \choose k}  (\cN+1)^k a_x = \sum_{k=0}^{j-1} {j \choose k}\, (\cN+1)^{\frac k2} a_x \cN^{\frac k2}.
\end{align*}
Therefore
\begin{equation}\label{eq:sp1-1}
\begin{split}
\frac{\rd}{\rd t} \langle \cU_N^{(M)} (t;0) \psi,& (\cN+1)^j
\cU_N^{(M)} (t;0) \psi \rangle \\ =  \; &\sum_{k=0}^{j-1} {j \choose k} (-1)^k \,
\text{Im} \, \int \rd x \rd y \,V(x-y) \, \ph_t (x) \ph_t
(y)  \\ &\hspace{.5cm} \times \langle \cU_N^{(M)} (t;0) \psi, \left ( \cN^{\frac k2} a_x^*a_y^*
(\cN+2)^{\frac k2}+
(\cN+1)^{\frac k2} a_x^*a_y^* (\cN+3)^{\frac k2}\right)
 \cU_N^{(M)} (t;0) \psi \rangle \\
&+\frac{2}{\sqrt{N}}\,  \sum_{k=0}^{j-1} {j \choose k} \,  \text{Im} \, \int \rd x \,
\\
&\hspace{.5cm}\times \langle \cU_N^{(M)} (t;0) \psi, a_x^* a (V(x-.)\ph_t) \chi (\cN \leq M)  (\cN+1)^{\frac k2} a_x \cN^{\frac k2} \cU_N^{(M)} (t;0) \psi \rangle.
\end{split}
\end{equation}
To control contributions from the first term we use bounds of the
form
\begin{equation*}
\begin{split} \Big|\int \rd x \rd y &V(x-y) \, \ph_t (x) \ph_t (y)\,
\langle \cU_N^{(M)} (t;0) \psi, (\cN+1)^{\frac k2}
a_x^* a_y^* (\cN+3)^{\frac k2} \cU_N^{(M)} (t;0) \psi \rangle \Big| \\
&\leq \int \rd x |\ph_t (x)| \| a_x (\cN+1)^{\frac k2} \cU_N^{(M)} (t;0)
\psi\| \, \| a^* (V(x-.)\ph_t)
(\cN+3)^{\frac k2}\cU_N^{(M)} (t;0) \psi\| \\
&\leq \const \,\sup_x  \left(\int V(x-y)^2 |\ph_t (y)|^2
\right)^{1/2} \| (\cN+3)^{\frac {k+1}2} \cU_N^{(M)} (t;0) \psi\|^2 \\
&\leq K \, \| (\cN+3)^{\frac {k+1}2} \cU_N^{(M)} (t;0) \psi\|^2\,.
\end{split}
\end{equation*}
Here we used that, by (\ref{eq:assump-coh}),
\begin{equation}\label{eq:bd} \sup_x \int \rd y V^2(x-y)
|\ph_t (y)|^2 \leq D \| \ph_t \|_{H^1}^2 \leq \const D \| \ph
\|^2_{H^1} \leq K
\end{equation}
is bounded uniformly in $t$ (as follows from (\ref{eq:H1ph})). Similar estimates
are applied to the term containing $\cN^{\frac k2} a_x^* a_y^* (\cN+2)^{\frac k2}$.

On the other hand, to control contributions arising from the second
integral on the r.h.s. of (\ref{eq:sp1-1}), we use estimates of the
form
\begin{equation*}
\begin{split}
\Big| \int \rd x &\, \langle \cU_N^{(M)} (t;0) \psi, a_x^* a
(V(x-.)\ph_t) \chi (\cN \leq M) (\cN+1)^{\frac k2} a_x \cN^{\frac k2}
\cU_N^{(M)} (t;0) \psi \rangle \Big| \\ &\leq \int \rd x \, \| a_x
(\cN+1)^{\frac k2} \cU_N^{(M)} (t;0) \psi\| \, \| a (V(x-.)\ph_t) \chi
(\cN \leq M)\|\,  \|a_x \cN^{\frac k2} \cU_N^{(M)} (t;0) \psi \| \\
&\leq M^{1/2} \sup_x \| V(x-.)\ph_t \| \,  \| \cN^{\frac {k+1}2}
 \cU_N^{(M)} (t;0) \psi \| \, \|\cN^{1/2}
(\cN+1)^{\frac k2} \cU_N^{(M)} (t;0) \psi \| \\ &\leq K M^{1/2} \|
(\cN+1)^{\frac {k+1}2}\cU_N^{(M)} (t;0) \psi\|^2\,.
\end{split}
\end{equation*}
This implies
\begin{equation*}
\begin{split}
\Big| \frac{\rd}{\rd t} \langle \cU_N^{(M)} (t;0) \psi, &(\cN+1)^j
\cU_N^{(M)} (t;0) \psi \rangle \Big| \\ &\leq  \; K (1+ \sqrt{M/N})\,
\sum_{k=0}^{j} {j \choose k}\,
\langle \cU_N^{(M)} (t;0) \psi, (\cN+3)^k \cU_N^{(M)} (t;0)
\psi\rangle\,\\ &\leq  \; 4^j \, K (1+ \sqrt{M/N}) \,
\langle \cU_N^{(M)} (t;0) \psi, (\cN+1)^j\, \cU_N^{(M)} (t;0)
\psi\rangle\, .
\end{split}
\end{equation*}
{F}rom Gronwall Lemma, we find (\ref{eq:step1cou}).
\end{proof}

\bigskip
\noindent
{\it Step 2. of the proof of Proposition \ref{prop:1cou}}

\subsection{Weak bounds on the $\cU_N$ dynamics}
To compare the evolution $\cU_N (t;s)$ with the cutoff evolution
$\cU^{(M)}_N (t;s)$, we first need some (very weak) a-priori bound on the
growth of the number of particle with respect to $\cU_N (t;s)$.

\begin{lemma}\label{lem:weak}
For arbitrary $t,s \in \bR$ and $\psi \in \cF$, we have
\begin{equation}\label{eq:step2cou1}
\langle \psi, \cU_N (t;s) \cN
\cU_N (t;s)^* \psi \rangle \leq 6 \langle \psi, (\cN+N+1) \psi \rangle\,.
\end{equation}
Moreover, for every $\ell \in \bN$, there exists a constant $C(\ell)$ such that
\begin{align}
\label{eq:step2cou}
\langle \psi, \cU_N (t;s) \cN^{2\ell}
&\cU_N (t;s)^* \psi \rangle \leq
C(\ell) \, \langle \psi, (\cN+N)^{2\ell}
\psi \rangle \\ \label{eq:step2coub}
\langle \psi, \cU_N (t;s) \cN^{2\ell+1}
&\cU_N (t;s)^* \psi \rangle \leq
C(\ell) \, \langle \psi, (\cN+N)^{2\ell+1} (\cN+1) \psi \rangle
\end{align}
for all $t,s \in \bR$ and $\psi\in \cF$.
\end{lemma}
\begin{proof}[Proof of Lemma \ref{lem:weak}.]
Eq. (\ref{eq:step2coub}) follows from (\ref{eq:step2cou}). In fact, assuming (\ref{eq:step2cou}) to hold true, we have
\begin{equation}
\begin{split}
\langle \psi, \cU_N (t;s) \cN^{2\ell+1} &\cU_N (t;s)^* \psi \rangle \\ \leq \; &\frac{1}{2N} \langle \psi, \cU_N (t;s) \cN^{2\ell+2} \cU_N (t;s)^* \psi \rangle + \frac{N}{2} \langle \psi, \cU_N (t;s) \cN^{2\ell} \cU_N (t;s)^* \psi \rangle \\ \leq \; &\frac{C(\ell+1)}{2N} \langle \psi, (\cN+N)^{2\ell+2} \psi \rangle + \frac{C(\ell) N}{2} \langle \psi, (\cN+N)^{2\ell} \psi \rangle \\ \leq \; &D(\ell) \, \langle \psi, (\cN+N)^{2\ell+1} (\cN+1) \psi \rangle
\end{split}
\end{equation}
for an appropriate constant $D(\ell)$.

\medskip

To prove (\ref{eq:step2cou1}) and (\ref{eq:step2cou}) we observe that, by (\ref{eq:GV}),
\begin{equation}\begin{split}\label{eq:UNU}
\cU_N^* &(t;s) \cN \cU_N (t;s) \\ &= \int \rd x \, \cU_N^* (t;s) a_x^*
a_x \cU_N (t;s) \\ &= \int \rd x \, W^* (\sqrt{N}\ph_s) e^{i\cH_N (t-s)}
(a_x^* - \sqrt{N} \overline{\ph}_t (x)) (a_x - \sqrt{N} \ph_t (x))
e^{-i\cH_N (t-s)} W(\sqrt{N} \ph_s) \\ &= W^* (\sqrt{N}\ph_s) \left( \cN - \sqrt{N}  e^{i\cH_N (t-s)} \phi (\ph_t)
e^{-i\cH_N (t-s)} + N  \right) W(\sqrt{N} \ph_s) \,.
\end{split}
\end{equation}
(Recall that $\phi (\ph) = a^* (\ph) + a (\ph) = \int \rd x (\ph
(x) a_x^* + \overline{\ph} (x)a_x )$). {F}rom Lemma \ref{lm:a-bd} and Lemma \ref{lm:coh}, we get
\begin{equation}
\begin{split}
\langle \psi, \cU_N^* (t;s) \cN \cU_N (t;s) \psi \rangle &\leq 2 \langle \psi, W^* (\sqrt{N}\ph_s) (\cN +N+1) W(\sqrt{N} \ph_s) \psi \rangle \\ &= 2 \langle \psi, (\cN + \sqrt{N} \phi (\ph_s) +2N+1) \psi \rangle \\ &\leq 6 \langle \psi, (\cN +N+1) \psi \rangle \, \end{split}
\end{equation}
which shows (\ref{eq:step2cou1}). To complete the proof of (\ref{eq:step2cou}), we define
\[ X_{t,s} = (\cN - \sqrt{N} e^{i\cH_N (t-s)} \phi (\ph_t) e^{-i\cH_N (t-s)} + N)\,.\]
Then, using the notation $\text{ad}_A (B) = [B,A]$, it is simple to prove that
there exists a constant $C>0$ such that
\begin{equation}\label{eq:Xtj} X_{t,s}^2 \leq C (N + \cN)^2 \qquad \text{and } \quad  \left| \text{ad}^{m}_{X_{t,s}} (\cN) \right| \leq C (N + \cN) \qquad \text{ for all $m \in \bN$.} \end{equation}
By induction it follows that, for every $\ell \in \bN$, there exist constants $D(\ell), C(\ell)$ with
\begin{equation}\label{eq:ind}
X_{t,s}^{\ell-1} (\cN + N)^2 X_{t,s}^{\ell-1} \leq D(\ell) (\cN+N)^{2\ell}
\qquad \text{and } \quad  X_{t,s}^{2\ell} \leq C(\ell)  (\cN+N)^{2\ell}\,. \end{equation}
In fact, for $\ell=1$ (\ref{eq:ind}) reduces to (\ref{eq:Xtj}). Assuming (\ref{eq:ind}) to hold for all $\ell <k$, we can prove it for $\ell=k$ by noticing that
\begin{equation}\label{eq:ind1}
\begin{split}
X_{t,s}^{k-1} (\cN + N)^2 X_{t,s}^{k-1}  &\leq 2 (\cN + N) X_{t,s}^{2k-2} (\cN+N) + 2 |[ X_{t,s}^{k-1}, \cN]|^2  \\ &\leq 2(\cN + N) X_{t,s}^{2k-2} (\cN+N) + 4^k \sum_{m=0}^{k-2} X_{t,s}^{m} \left|\text{ad}_{X_{t,s}}^{k-1-m} (\cN) \right|^2 X_{t-s}^m  \\
&\leq 2 (\cN + N) X_{t,s}^{2k-2} (\cN+N) + 4^k C  \sum_{m=0}^{k-2} X_{t,s}^{m} (\cN+N)^2 X_{t-s}^m
\\ &\leq D(k) \, (\cN+N)^{2k}
\end{split}
\end{equation}
for an appropriate constant $D(k)$, and that, by (\ref{eq:Xtj}) and (\ref{eq:ind1}),
\[ X_{t,s}^{2k} \leq C X_{t,s}^{k-1} (\cN +N)^2 X_{t,s}^{k-1} \leq C D(k) (\cN+N)^{2k} = C(k) (\cN+N)^{2k} \,. \] In (\ref{eq:ind1}), we used the commutator expansion
\[ [A^n, B] = \sum_{m=0}^{n-1} {n \choose m} A^m \text{ad}^{n-m}_A (B)\,  \] in the second line, the bound (\ref{eq:Xtj}) in the third line, and the induction assumption in the last line.

{F}rom (\ref{eq:UNU}) and (\ref{eq:ind}), we obtain that
\begin{equation}\label{eq:WNW}
\begin{split}
\langle \psi, \cU_N (t;s) \cN^{2\ell} \cU_N (t;s)^* \psi \rangle &= \langle W (\sqrt{N} \ph_s) \psi, X_{t,s}^{2\ell} W (\sqrt{N} \ph_s) \psi \rangle \\ & \leq C(\ell)
\langle W (\sqrt{N} \ph_s) \psi, (\cN+N)^{2\ell} W (\sqrt{N} \ph_s) \psi \rangle \\ & = C(\ell)
\langle \psi, (\cN+\sqrt{N} \phi (\ph_s) + 2N)^{2\ell} \psi \rangle\,.
\end{split}
\end{equation}
Analogously to (\ref{eq:ind}), it is possible to prove that, for every $\ell \in \bN$, there exists a constant $C(\ell)$ with \[
(\cN+\sqrt{N} \phi (\ph_s) + 2N)^{2\ell} \leq C(\ell) (\cN + N)^{2\ell} \,. \] Eq.~(\ref{eq:step2cou}) follows therefore from (\ref{eq:WNW}).
\end{proof}

\bigskip

\noindent
{\it Step 3. of the proof of Proposition \ref{prop:1cou}}

\subsection{Comparison of the $\cU_N$ and $\cU_N^{(M)}$ dynamics}
\begin{lemma}\label{lem:comp}
For every $j\in \bN$ there exist constants $C(j),K(j)$ (depending only on $j$, on $\| \ph
\|_{H^1}$ and on the constant $D$ in (\ref{eq:assump-coh})) such
that
\begin{equation}
\begin{split}
\label{eq:step31cou} \Big| \langle \cU_N (t;s) \psi, \cN^j \Big(
\cU_N (t;s) - &\cU_N^{(M)} (t;s)\Big) \psi \rangle \Big| \\ &\leq C(j) \frac{(N/M)^j \, \| (\cN+1)^{j+1} \psi \|^2}{(1+\sqrt{M/N})} \, \exp \left(K(j) (1+ \sqrt{M/N}) |t-s| \right)\,
\end{split}
\end{equation}
and
\begin{equation}
\begin{split}
\label{eq:step32cou}
\Big| \langle \cU^{(M)}_N (t;s) \psi, \cN^j
\Big( \cU_N (t;s) - &\cU_N^{(M)} (t;s)\Big) \psi \rangle \Big|
\\ &\leq C \frac{\| (\cN+1)^j \psi \|^2}{M^j (1+\sqrt{M/N})} \exp \left( K(j) (1 + \sqrt{M/N}) |t-s| \right)\, ,
\end{split}
\end{equation}
for all $\psi \in \cF$ and for all $t,s \in \bR$.
\end{lemma}
\begin{proof}[Proof of Lemma \ref{lem:comp}.]
To simplify the notation we consider the case $s=0$ and $t >0$ (but the other cases can be treated identically). To prove (\ref{eq:step31cou}), we expand the difference of the two
evolutions:
\begin{equation}\label{eq:sp3-1}
\begin{split}
\langle \cU_N (t;0) &\psi, \cN^j \left( \cU_N (t;0) - \cU_N^{(M)}
(t;0)\right) \psi \rangle \\ = \; &\langle \cU_N (t;0) \psi, \cN^j
\cU_N (t;0) \left( 1 - \cU_N (t;0)^* \cU_N^{(M)} (t;0)\right) \psi
\rangle \\ = \; &-i\int_0^t \rd s \; \langle \cU_N (t;0) \psi, \cN^j
\cU_N (t;s) \left(\cL_N (s) - \cL^{(M)}_N (s) \right) \cU_N^{(M)}
(s;0) \psi \rangle \\
= \; &-\frac{i}{\sqrt{N}} \int_0^t \rd s \int \rd x \rd y  V(x-y)  \\
&\hspace{.2cm} \times \langle \cU_N (t;0) \psi, \cN^j \cU_N (t;s)
a^*_x \left( \overline{\ph}_t (y) a_y \chi (\cN > M) + \ph_t (y)
\chi (\cN
>M) a_y^* \right) a_x \cU_N^{(M)}
(s;0) \psi \rangle \\
= \; &-\frac{i}{\sqrt{N}}\int_0^t \rd s \int \rd x  \langle a_x
\cU_N (t;s)^* \cN^j \cU_N (t;0) \psi, a (V(x-.)\ph_t) \chi (\cN
> M) a_x \cU_N^{(M)} (s;0) \psi \rangle \\
&-\frac{i}{\sqrt{N}}\int_0^t \rd s \int \rd x  \langle a_x \cU_N
(t;s)^* \cN^j \cU_N (t;0) \psi, \chi (\cN >M) a^* (V(x-.)\ph_t) a_x
\cU_N^{(M)} (s;0) \psi \rangle\,.
\end{split}
\end{equation}
Hence
\begin{equation*}
\begin{split}
\Big| \langle \cU_N (t;0) &\psi, \cN^j \left( \cU_N (t;0) -
\cU_N^{(M)} (t;0)\right) \psi \rangle \Big| \\ \leq
 \; &\frac{1}{\sqrt{N}} \int_0^t \rd s \int \rd x  \| a_x \cU_N (t;s)^* \cN^j
\cU_N (t;0) \psi\| \, \| a (V(x-.)\ph_t)  a_x \chi (\cN > M+1)
\cU_N^{(M)} (s;0) \psi \|  \\
&+\frac{1}{\sqrt{N}}\int_0^t \rd s \int \rd x  \| a_x \cU_N (t;s)^*
\cN^j \cU_N (t;0) \psi\| \| a^* (V(x-.)\ph_t) a_x \chi (\cN
>M) \cU_N^{(M)} (s;0) \psi \| \\
\leq \; &\frac{1}{\sqrt{N}} \sup_x \| V(x-.)\ph_t \| \int_0^t \rd s
\int \rd x  \| a_x \cU_N (t;s)^* \cN^j \cU_N (t;0) \psi \| \, \\
&\hspace{6cm} \times \| a_x \cN^{1/2} \chi (\cN > M+1)
\cU_N^{(M)} (s;0) \psi \|  \\
&+\frac{1}{\sqrt{N}} \sup_x \| V(x-.)\ph_t \| \int_0^t \rd s \int
\rd x  \| a_x \cU_N (t;s)^* \cN^j \cU_N (t;0) \psi \| \\
&\hspace{6cm} \times \| a_x \cN^{1/2} \chi (\cN
>M) \cU_N^{(M)} (s;0) \psi \| \\
\leq \; &\frac{C}{\sqrt{N}} \int_0^t \rd s  \|\cN^{1/2} \cU_N
(t;s)^* \cN^j \cU_N (t;0) \psi\| \, \|  \cN \chi (\cN > M)
\cU_N^{(M)} (s;0) \psi \|
\end{split}
\end{equation*}
where we used (\ref{eq:bd}) once again. {F}rom Lemma \ref{lem:weak}, we obtain
\begin{equation}
\begin{split}
\|\cN^{1/2} \cU_N (t;s)^* \cN^j \cU_N (t;0) \psi\|^2 &= \langle \cN^j \cU_N (t;0) \psi, \cU (t;s) \cN \cU_N (t;s)^* \cN^j \cU_N (t;0) \psi \rangle \\ &\leq 6 \langle \cN^j \cU_N (t;0) \psi, (\cN+N+1) \cN^j \cU_N (t;0) \psi \rangle \\ & \leq C(j) \langle \psi, (\cN+N)^{2j+1} (\cN+1) \psi \rangle \\ & \leq C(j) N^{2j+1} \langle \psi, (\cN+1)^{2j+2} \psi \rangle\,.
\end{split}
\end{equation}
Therefore, using the inequality $\chi (\cN > M) \leq (\cN/M)^{2j}$, we obtain
\begin{equation*}
\begin{split}
\Big| \langle \cU_N (t;0) \psi, &\cN^j \left( \cU_N (t;0) -
\cU_N^{(M)} (t;0)\right) \psi \rangle \Big| \\ &\leq C(j) N^j \, \|(\cN+1)^{j+1} \psi\|
\int_0^t \rd s \; \langle \cU_N^{(M)} (s;0) \psi, \cN^2 \chi (\cN
>M) \cU_N^{(M)} (s;0) \psi \rangle^{1/2} \\ &\leq C(j) N^j \, \|(\cN+1)^{j+1} \psi\|   \int_0^t \rd s \; \langle \cU_N^{(M)} (s;0) \psi, \frac{\cN^{2j+2}}{M^{2j}}
\cU_N^{(M)} (s;0) \psi \rangle^{1/2}\, . \end{split}
\end{equation*}
Finally, from (\ref{eq:step1cou}), we conclude that
\begin{equation*}
\begin{split}
\Big| \langle \cU_N (t;0) \psi, \cN^j &\left( \cU_N (t;0) -
\cU_N^{(M)} (t;0)\right) \psi \rangle \Big| \\ &\leq C(j) (N/M)^j
\|(\cN+1)^{j+1} \psi\|^2 \int_0^t \rd s \; \exp \, (K(j) \, s \, (1+ \sqrt{M/N}))
\\ & \leq C(j) \frac{(N/M)^j \|(\cN+1)^{j+1} \psi\|^2}{1+\sqrt{M/N}} \, \exp \, (K(j) \, t\, (1+ \sqrt{M/N}))\,.
\end{split}
\end{equation*}

\medskip

To prove (\ref{eq:step32cou}), we proceed similarly; analogously to
(\ref{eq:sp3-1}) we find
\begin{equation*}
\begin{split}
\langle \cU^{(M)}_N (t;0) &\psi, \cN^j \left( \cU_N (t;0) -
\cU_N^{(M)}
(t;0)\right) \psi \rangle \\
= \; &-\frac{i}{\sqrt{N}}\int_0^t \rd s \int \rd x  \langle a_x
\cU_N (t;s)^* \cN^j \cU^{(M)}_N (t;0) \psi, a (V(x-.)\ph_t) \chi
(\cN > M) a_x \cU_N^{(M)} (s;0) \psi \rangle \\
&-\frac{i}{\sqrt{N}}\int_0^t \rd s \int \rd x  \langle a_x \cU_N
(t;s)^* \cN^j \cU^{(M)}_N (t;0) \psi, \chi (\cN >M) a^*
(V(x-.)\ph_t) a_x \cU_N^{(M)} (s;0) \psi \rangle
\end{split}
\end{equation*}
and thus
\begin{equation}
\begin{split}
\Big| \langle \cU^{(M)}_N (t;0) &\psi, \cN^j \left( \cU_N (t;0) -
\cU_N^{(M)} (t;0)\right) \psi \rangle \Big| \\
\leq \; &\frac{C}{\sqrt{N}} \int_0^t \rd s  \|\cN^{1/2} \cU_N
(t;s)^* \cN^j \cU_N^{(M)} (t;0) \psi \| \, \|  \cN \chi (\cN > M)
\cU_N^{(M)} (s;0) \psi \|\,.
\end{split}
\end{equation}
Again, applying (\ref{eq:step2cou}) and (\ref{eq:step1cou}) we find
\begin{equation*}
\Big| \langle \cU^{(M)}_N (t;0) \psi, \cN^j \left( \cU_N (t;0) -
\cU_N^{(M)} (t;0)\right) \psi \rangle \Big| \leq C \frac{\| (\cN+1)^{j+1} \psi \|^2 }{M^j(1+\sqrt{M/N})} \exp
\, (K(j) \, t \, (1 + \sqrt{M/N}))\,.
\end{equation*}
\end{proof}

\bigskip
\noindent
{\it Step 4. Conclusion of the proof  of Proposition \ref{prop:1cou}}

\medskip
{F}rom (\ref{eq:step31cou}),
(\ref{eq:step32cou}) and (\ref{eq:step1cou}) we obtain, choosing
$M=N$,
\begin{equation*}
\begin{split}
\langle \cU_N (t;s) \psi, \cN^j \cU_N (t;s) \psi \rangle = \; &
\langle \cU_N (t;s) \psi, \cN^j ( \cU_N (t;s) -\cU_N^{(M)} (t;s)
) \psi \rangle \\ &+ \langle (\cU_N (t;s) - \cU_N^{(M)} (t;s))
\psi, \cN^j \, \cU_N^{(M)} (t;s) \psi \rangle \\&+ \langle
\cU_N^{(M)} (t;s) \psi, \cN^j \, \cU_N^{(M)} (t;s) \psi \rangle \\
\leq \; & C (j) \| (\cN+1)^{j+1} \psi \|^2 e^{K(j) |t-s|}\,.
\end{split}\end{equation*}
\end{proof}
\subsection{Approximate dynamics $\wt \cU_N (t;s)$}
We now consider the dynamics $\wt \cU_N (t;s)$, defined in \eqref{eq:wtcUN} by
$$
i\partial_t \wt \cU_N (t;s) = \wt\cL_N (t)\, \wt\cU_N (t;s) \qquad
\text{with} \quad \wt\cU_N (s;s) = 1
$$ with the
time-dependent generator
\begin{equation}
\begin{split}
\wt \cL_N (t) = & \int \rd x \, \nabla_x a^*_x \nabla_x a_x + \int
\rd x \, \left(V *|\ph_t|^2 \right) (x) \, a^*_x a_x + \int \rd x
\rd y \, V(x-y) \overline{\ph_t} (x) \ph_t (y) a^*_y a_x \\
&+ \frac{1}{2} \int \rd x \rd y \, V(x-y) \left( \ph_t (x) \ph_t (y)
a^*_x a^*_y +
\overline{\ph_t} (x) \overline{\ph_t} (y) a_x a_y \right) \\
&+\frac{1}{2N} \int \rd x \rd y \, V(x-y) \, a^*_x a^*_y a_y a_x \, .
\end{split}
\end{equation}
\begin{lemma}\label{lm:2cou}
There exists a constant $K>0$, only depending
on $\| \ph \|_{H^1}$ and on the constant $D$ appearing in
(\ref{eq:assump-coh}), such that
\begin{equation} \label{eq:lm2cou} \langle \wt \cU_N (t;0) \Omega,
\cN^3 \wt \cU_N (t;0) \Omega \rangle \leq e^{Kt} \, .\end{equation}
\end{lemma}

\begin{proof}
We compute the derivative
\begin{equation*}
\begin{split}
\frac{\rd}{\rd t} \langle &\wt \cU_N (t;0) \Omega, (\cN+1)^3 \wt
\cU_N (t;0) \Omega \rangle \\ &= \langle \wt \cU_N (t;0) \Omega,
[i\wt\cL_N (t), (\cN+1)^3 ] \wt \cU_N (t;0) \Omega \rangle \\ &=
2\text{Im} \int \rd x \rd y V(x-y) \ph_t (x) \ph_t (y) \langle \wt
\cU_N (t;0) \Omega, [ a_x^* a_y^* , (\cN+1)^3 ] \wt \cU_N (t;0)
\Omega \rangle
\\ &= 4\text{Im} \int
\rd x \rd y V(x-y) \ph_t (x) \ph_t (y) \langle \wt \cU_N (t;0)
\Omega, \\ &\hspace{2cm} \left( a_x^* a_y^*  (\cN+1)^2  + (\cN+1)
a_x^* a_y^* (\cN+1) + (\cN+1)^2
a_x^* a_y^* \right) \wt \cU_N (t;0) \Omega \rangle \\
 &= 4\text{Im} \int
\rd x \rd y V(x-y) \ph_t (x) \ph_t (y) \langle \wt \cU_N (t;0)
\Omega,
\\ &\hspace{2cm} \left( (\cN-1) a_x^* a_y^*  (\cN+1)  + (\cN+1) a_x^* a_y^*
(\cN+1) + (\cN+1)
a_x^* a_y^* (\cN+3) \right) \wt \cU_N (t;0) \Omega \rangle\\
 &= 4\text{Im} \int
\rd x \rd y V(x-y) \ph_t (x) \ph_t (y) \langle \wt \cU_N (t;0)
\Omega, \left( 3(\cN+1) a_x^* a_y^*  (\cN+1)  - 4 a_x^* a_y^*
\right) \wt \cU_N (t;0) \Omega \rangle\,.
\end{split}
\end{equation*}
Therefore
\begin{equation*}
\begin{split}
\frac{\rd}{\rd t} \langle \wt \cU_N (t;0) &\Omega, (\cN+1)^3 \wt
\cU_N (t;0) \Omega \rangle \\ = \; &12 \text{Im} \int \rd x \ph_t
(x) \langle a_x (\cN+1) \wt \cU_N (t;0) \Omega, a^* (V(x-.)\ph_t)
(\cN+1) \wt \cU_N (t;0) \Omega \rangle \\ \; &-16 \text{Im} \int \rd
x \ph_t (x) \langle a_x \wt \cU_N (t;0) \Omega, a^* (V(x-.)\ph_t)
\wt \cU_N (t;0) \Omega \rangle\,.
\end{split}
\end{equation*}
Taking the absolute value, we find
\begin{equation*}
\begin{split}
\Big| \frac{\rd}{\rd t} \langle \wt \cU_N (t;0) &\Omega, (\cN+1)^3
\wt \cU_N (t;0) \Omega \rangle \Big| \\ \leq \; &12 \int \rd x
|\ph_t (x)| \| a_x (\cN+1) \wt \cU_N (t;0) \Omega \| \, \| a^*
(V(x-.)\ph_t) (\cN+1) \wt \cU_N (t;0) \Omega \| \\ \; &+16  \int \rd
x |\ph_t (x)| \|a_x \wt \cU_N (t;0) \Omega\| \, \| a^* (V(x-.)\ph_t)
\wt \cU_N (t;0) \Omega \| \\ \leq \; &28 \sup_x \| V(x-.) \ph_t \|
\| (\cN+1)^{3/2} \wt \cU_N (t;0)\Omega \|^2 \\ \leq \; &C \, \langle
\wt \cU_N (t;0) \Omega, (\cN+1)^3 \wt \cU_N (t;0) \Omega \rangle\,.
\end{split}
\end{equation*}
Applying Gronwall Lemma, we obtain (\ref{eq:lm2cou}).
\end{proof}
\subsection{Comparison of the $\cU_N$ and $\wt\cU_N$ dynamics}
The final step in the proof of Theorem \ref{thm:coh} is the comparison of
evolutions generated by $\cU_N$ and $\wt\cU_N$.
\begin{lemma}\label{lm:3cou}
Let the evolutions $\cU_N (t;s)$ and $\wt\cU_N (t;s)$ be
defined as in (\ref{eq:cUN}) and (\ref{eq:wtcUN}), respectively.
Then there exist constants $C,K>0$, only depending on $\| \ph
\|_{H^1}$ and on the constant $D$ in (\ref{eq:assump-coh}), such
that
\begin{equation}\label{eq:lm3cou} \left\| \left(\cU_N (t;0) - \wt \cU_N
(t;0) \right) \Omega \right\| \leq \frac{C}{\sqrt{N}} \, e^{Kt}\,.
\end{equation}
\end{lemma}

\begin{proof}
We write
\begin{equation*}
\begin{split}
\Big(\cU_N (t;0) - &\wt \cU_N (t;0) \Big) \Omega\\  = \; &\cU_N
(t;0) \left(1 - \cU_N (t;0)^* \wt \cU_N (t;0) \right) \Omega \\ = \;
&- i\int_0^t \rd s \, \cU_N (t;s) \left(\cL_N (s) - \wt\cL_N (s)
\right) \wt\cU_N (s;0) \Omega \\ = \; &- \frac{i}{\sqrt{N}} \int_0^t
\rd s \, \int \rd x \rd y \, V(x-y) \, \cU_N (t;s) a^*_x \left(
\ph_t (y) a_y^* + \overline{\ph}_t (y) a_y \right) a_x \wt\cU_N
(s;0) \Omega \\= \; &- \frac{i}{\sqrt{N}} \int_0^t \rd s \, \int \rd
x  \, \cU_N (t;s) a^*_x \phi (V(x-.) \ph_t) a_x \wt \cU_N (s;0)
\Omega\,.
\end{split}
\end{equation*}
Hence
\begin{equation}\label{eq:lm31}
\begin{split}
\left\| \left(\cU_N (t;0) - \wt \cU_N (t;0) \right) \Omega \right\|
\leq \frac{1}{\sqrt{N}} \int_0^t \rd s \, \left\| \int \rd x  \,
a^*_x \phi (V(x-.) \ph_t) a_x \wt \cU_N (s;0) \Omega \right\|\,.
\end{split}
\end{equation}
Next, we observe that
\begin{equation*}
\begin{split}
\Big\| \int \rd x  \, a^*_x &\phi (V(x-.) \ph_t) a_x \wt \cU_N (s;0)
\Omega \Big\|^2 \\ = \; &\int \rd y \rd x \langle a_y \wt \cU_N
(s;0) \Omega, \phi (V(y-.)\ph_t) a_y a^*_x \phi (V(x-.) \ph_t) a_x
\wt \cU_N (s;0) \Omega \rangle \\ = \; &\int \rd y \rd x \langle a_y
\wt \cU_N (s;0) \Omega, \phi (V(y-.)\ph_t) a^*_x a_y\phi (V(x-.)
\ph_t) a_x \wt \cU_N (s;0) \Omega \rangle \\ &+\int \rd x \langle
a_x \wt \cU_N (s;0) \Omega, \phi (V(x-.)\ph_t) \phi (V(x-.) \ph_t)
a_x \wt \cU_N (s;0) \Omega \rangle \\ = \; &\int \rd y \rd x \langle
a_y \wt \cU_N (s;0) \Omega, \left( a_x^* \phi (V(y-.)\ph_t) + V(y-x)
\overline{\ph}_t (x) \right) \\& \hspace{5cm} \times  \left( \phi
(V(x-.) \ph_t) a_y + V(x-y) \ph_t (y) \right)
a_x \wt \cU_N (s;0) \Omega \rangle \\
&+\int \rd x \langle a_x \wt \cU_N (s;0) \Omega, \phi (V(x-.)\ph_t)
\phi (V(x-.) \ph_t) a_x \wt \cU_N (s;0) \Omega \rangle\,.
\end{split}\end{equation*}
Therefore, we have
\begin{equation*}
\begin{split}
\Big\| \int \rd x  \, a^*_x &\phi (V(x-.) \ph_t) a_x \wt \cU_N (s;0)
\Omega \Big\|^2 \\
= \; &\int \rd y \rd x \langle a_x a_y \wt \cU_N (s;0) \Omega, \phi
(V(y-.)\ph_t) \phi (V(x-.) \ph_t) a_y
a_x \wt \cU_N (s;0) \Omega \rangle \\
&+\int \rd y \rd x V(x-y) \overline{\ph}_t (x) \langle a_y \wt \cU_N
(s;0) \Omega, \phi (V(x-.) \ph_t) a_y a_x \wt \cU_N (s;0) \Omega
\rangle
\\
&+\int \rd y \rd x V(x-y) \ph_t (y) \langle a_x a_y \wt \cU_N (s;0)
\Omega, \phi (V(y-.) \ph_t) a_x \wt \cU_N (s;0) \Omega \rangle\\
&+\int \rd y \rd x V(x-y)^2 \overline{\ph}_t (x) \ph_t (y) \langle
a_y \wt
\cU_N (s;0) \Omega, a_x \wt \cU_N (s;0) \Omega \rangle \\
&+\int \rd x \langle a_x \wt \cU_N (s;0) \Omega, \phi (V(x-.)\ph_t)
\phi (V(x-.) \ph_t) a_x \wt \cU_N (s;0) \Omega \rangle\,.
\end{split}
\end{equation*}
It follows that
\begin{equation*}
\begin{split}
\Big\| \int &\rd x  \, a^*_x \phi (V(x-.) \ph_t) a_x \wt \cU_N (s;0)
\Omega \Big\|^2 \\
\leq \; & \sup_{x} \| V(x-.)\ph_t \|^2 \, \int \rd y \rd x \|
(\cN+2)^{1/2} a_x a_y \wt \cU_N (s;0) \Omega\|^2 \, \\
&+\sup_{x} \| V(x-.)\ph_t \| \;  \int \rd y \rd x |V(x-y)| |\ph_t
(x)| \| (\cN+1)^{1/2} a_y \wt \cU_N (s;0) \Omega \| \| a_y a_x \wt
\cU_N (s;0) \Omega \|
\\
&+ \sup_{y} \| V(y-.)\ph_t \| \; \int \rd y \rd x |V(x-y)| |\ph_t
(y)| \| a_x a_y \wt \cU_N (s;0)
\Omega\|  \| (\cN+1)^{1/2} a_x \wt \cU_N (s;0) \Omega \|\\
&+\int \rd y \rd x V(x-y)^2 |\ph_t (x)||\ph_t (y)| \| a_y \wt
\cU_N (s;0) \Omega \| \| a_x \wt \cU_N (s;0) \Omega \| \\
&+ \sup_x \| V(x-. )\ph_t\|^2 \, \int \rd x \| (\cN+1)^{1/2} a_x \wt
\cU_N (s;0) \Omega\|^2\,.
\end{split}
\end{equation*}
Using (\ref{eq:bd}), we obtain
\begin{equation*}
\begin{split}
\Big\| \int &\rd x  \, a^*_x \phi (V(x-.) \ph_t) a_x \wt \cU_N (s;0)
\Omega \Big\|^2 \\
\leq \; & C \, \int \rd y \rd x \|
 a_x a_y \cN^{1/2} \wt \cU_N (s;0) \Omega \|^2 \, \\
&+C \;  \left(\int \rd y \rd x |V(x-y)|^2 |\ph_t (x)|^2 \|
 a_y \cN^{1/2} \wt \cU_N (s;0) \Omega \|^2 \right)^{1/2} \left(
\int \rd x \rd y \| a_y a_x \wt \cU_N (s;0) \Omega \|^2
\right)^{1/2}
\\
&+ C \left(\int \rd y \rd x \| a_x a_y \wt \cU_N (s;0) \Omega\|^2
\right)^{1/2}  \left( \int \rd y \rd x |V(x-y)|^2 |\ph_t (y)|^2
\| a_x \cN^{1/2} \wt \cU_N (s;0) \Omega \|^2 \right)^{1/2} \\
&+ \int \rd y \rd x V(x-y)^2 |\ph_t (x)|^2 \| a_y \wt \cU_N (s;0)
\Omega \|^2 \\ &+ C \int \rd x \| a_x \cN^{1/2} \wt \cU_N (s;0)
\Omega\|^2\,.
\end{split}
\end{equation*}
{F}rom
\[\int \rd y \rd x |V(x-y)|^2 |\ph_t (y)|^2
\| a_x \psi \|^2 \leq \left( \sup_x \int \rd y V(x-y)^2 |\ph_t
(y)|^2 \right) \| \cN^{1/2} \psi \|^2 \leq C \|\cN^{1/2} \psi \|^2
\] we thus find
\begin{equation*}
\begin{split}
\Big\| \int \rd x  \, a^*_x &\phi (V(x-.) \ph_t) a_x \wt \cU_N (s;0)
\Omega \Big\|^2 \leq C \| (\cN+1)^{3/2} \wt \cU_N (t;0) \Omega
\|^2\,.
\end{split}
\end{equation*}
Inserting the last bound in (\ref{eq:lm31}) and using the result of
Lemma \ref{lm:2cou} we obtain (\ref{eq:lm3cou}).
\end{proof}
This concludes the proof of Theorem \ref{thm:coh}.
\subsection{Discussion}\label{subs:disc}
As mentioned in the introduction, our approach to the study of the mean field limit of
the $N$-body Schr\"odinger dynamics mirrors that
used by Hepp and Ginibre-Velo in \cite{Hepp,GV}
in the study of the semi-classical limit of quantum many-boson systems.
In the language of the mean field limit,
the main result obtained by Hepp
(for smooth potentials) and by Ginibre and Velo (for singular
potentials) was the convergence of the fluctuation dynamics $\cU_N
(t;s)$ (defined in (\ref{eq:cUN})) to a limiting $N$-independent
dynamics $\cU(t;s)$ in the sense that
\begin{equation}\label{eq:HGV} s-\lim_{N \to \infty} \cU_N (t;s) =
\cU (t;s) \, \end{equation}  for all fixed $t$ and $s$. Here the
limiting dynamics $\cU(t;s)$ is defined by
\[ i\partial_t \cU (t;s) = \cL (t) \cU (t;s) \qquad
\text{with } \cU (s;s) = 1 \] and with generator
\begin{equation}\label{eq:cL}
\begin{split}
\cL (t) = & \int \rd x \, \nabla_x a^*_x \nabla_x a_x + \int \rd x
\, \left(V *|\ph_t|^2 \right) (x) \, a^*_x a_x + \int \rd x
\rd y \, V(x-y) \overline{\ph_t} (x) \ph_t (y) a^*_y a_x \\
&+ \frac{1}{2} \int \rd x \rd y \, V(x-y) \left( \ph_t (x) \ph_t (y)
a^*_x a^*_y + \overline{\ph_t} (x) \overline{\ph_t} (y) a_x a_y
\right)
\end{split}
\end{equation}
The convergence (\ref{eq:HGV}) does not give any information about
the convergence of the one-particle marginal $\Gamma_{N,t}^{(1)}$,
associated with the evolution of the coherent initial state, to the
orthogonal projection $|\ph_t \rangle \langle \ph_t|$. The
definition of the marginal density
$\Gamma_{N,t}^{(1)}$ involves {\it unbounded}
creation and annihilation operators. This also explains why the
derivation of the bound (\ref{eq:tr-1}) in Theorem \ref{thm:coh} is
in general more complicated than the proof of the convergence
(\ref{eq:HGV}). The proof of (\ref{eq:HGV}) requires control of the
growth of the expectation of powers of the number of particle
operator $\cN$ {\it only} with respect to the {\it limiting dynamics}. To prove
(\ref{eq:tr-1}), on the other hand, we  need to control the
growth of the expectation of $\cN$ with respect to the $N$-dependent
fluctuation dynamics $\cU_N (t;s)$.

\section{Time evolution of factorized states}\label{sec:fact}
\setcounter{equation}{0}

This section is devoted to the proof of Theorem \ref{thm:fact}. The
main idea in the proof is that we can write the factorized
$N$-particle state $\psi_N = \ph^{\otimes N}$ (whose evolution is
considered in Theorem \ref{thm:fact}) as a linear combination of
coherent states, whose dynamics can be studied using the results of
Section \ref{sec:coh}.

\begin{proof}[Proof of Theorem \ref{thm:fact}]
We start by writing $\psi_N = \ph^{\otimes N}$ or, more precisely,
the sequence \[ \{ 0, 0, \dots, 0, \psi_N, 0, 0, \dots \} =
\frac{(a^* (\ph))^N}{\sqrt{N!}} \Omega \in \cF\] as a linear
combination of coherent states. While it is always possible in principle our goal
is to represent $\psi_N$ with the {\it least} number of coherent states.
\begin{lemma}\label{lem:rep}

We have the following representation.
\begin{equation}\label{eq:cohpsi}
\frac{(a^* (\ph))^N}{\sqrt{N!}} \Omega = d_N \int_0^{2\pi} \frac{\rd
\theta}{2\pi} \; e^{i\theta N} W ( e^{-i\theta} \sqrt{N} \ph) \Omega
\end{equation} with the constant
\begin{equation}
d_N = \frac{\sqrt{N!}}{N^{N/2} e^{-N/2}} \simeq N^{1/4}\,.
\end{equation}
\end{lemma}
\begin{proof}
To prove the representation (\ref{eq:cohpsi}) observe that, from
(\ref{eq:coh}) and since $\| \ph \| =1$, \begin{equation}
\begin{split} \int_0^{2\pi} \frac{\rd \theta}{2\pi} \; e^{i\theta N} W
(e^{-i\theta} \sqrt{N} \ph) \Omega &= e^{-N/2} \sum_{j=1}^{\infty}
N^{j/2} \left( \int \frac{\rd \theta}{2\pi} \; e^{i\theta (N-j)}
\right) \frac{(a^* (\ph))^j}{j!} \Omega \\
&= \frac{e^{-N/2} N^{N/2}}{\sqrt{N!}} \, \frac{(a^*
(\ph))^N}{\sqrt{N!}} \Omega \, .
\end{split}
\end{equation}
\end{proof}
The kernel of the one-particle density $\gamma_{N,t}^{(1)}$
associated with the solution of the Schr\"odinger equation
\[ e^{-it\cH_N} \frac{(a^* (\ph))^N}{\sqrt{N!}} \Omega \]
is given by (see (\ref{eq:margi}))
\begin{equation}
\begin{split}
\gamma^{(1)}_{N,t} (x;y) =\; &\frac{1}{N} \left\langle \frac{(a^*
(\ph))^N}{\sqrt{N!}} \Omega, e^{i \cH_N t} a_y^* a_x e^{-i\cH_N t}
\frac{(a^* (\ph))^N}{\sqrt{N!}} \Omega \right\rangle
\\ = \; & \frac{d^2_N}{N} \int_0^{2\pi} \frac{\rd \theta_1}{2\pi} \int_0^{2\pi}
\frac{\rd \theta_2}{2\pi} \, e^{-i\theta_1 N} e^{i\theta_2 N}
\langle W(e^{-i\theta_1} \sqrt{N} \ph) \Omega, a_y^* (t) a_x (t)
W(e^{-i\theta_2} \sqrt{N} \ph) \Omega \rangle
\end{split}
\end{equation}
where we introduced the notation $a_x (t) = e^{i\cH_N t} a_x
e^{-i\cH_N t}$. Next, we expand
\begin{equation}\label{eq:onepart}
\begin{split}
\gamma^{(1)}_{N,t} (x;y) = \; & \frac{d^2_N}{N} \int_0^{2\pi}
\frac{\rd \theta_1}{2\pi} \int_0^{2\pi} \frac{\rd \theta_2}{2\pi} \,
e^{-i\theta_1 N} e^{i\theta_2 N}  \left\langle W(e^{-i\theta_1}
\sqrt{N} \ph) \Omega, \, \left( a_y^*
(t) - e^{i\theta_1} \sqrt{N} \overline{\ph}_t (y) \right) \right. \\
&\hspace{3cm} \left. \times \left( a_x (t) - e^{-i\theta_2} \sqrt{N}
\ph_t (x) \right) \, W(e^{-i\theta_2} \sqrt{N} \ph) \Omega
\right\rangle
\\ & + \frac{d^2_N \, \overline{\ph}_t (y)}{\sqrt{N}} \int_0^{2\pi}
\frac{\rd \theta_1}{2\pi} \int_0^{2\pi} \frac{\rd \theta_2}{2\pi} \,
e^{-i\theta_1 (N-1)} e^{i\theta_2 N} \\ &\hspace{3cm} \times
\left\langle W(e^{-i\theta_1} \sqrt{N} \ph) \Omega, \left( a_x (t) -
e^{-i\theta_2} \sqrt{N} \ph_t (x) \right)W(e^{-i\theta_2} \sqrt{N}
\ph) \Omega \right\rangle
\\& + \frac{d^2_N \, \ph_t (x)}{\sqrt{N}} \int_0^{2\pi} \frac{\rd
\theta_1}{2\pi} \int_0^{2\pi} \frac{\rd \theta_2}{2\pi} \,
e^{-i\theta_1 N} e^{i\theta_2 (N-1)} \\ &\hspace{3cm} \times
\left\langle W(e^{-i\theta_1} \sqrt{N} \ph) \Omega, \left( a^*_y (t)
- e^{i\theta_1} \sqrt{N} \overline{\ph}_t (y) \right) \,
W(e^{-i\theta_2} \sqrt{N} \ph) \Omega \right\rangle \\ &+ d_N^2 \,
\ph_t (x) \overline{\ph}_t (y) \int_0^{2\pi} \frac{\rd
\theta_1}{2\pi} \int_0^{2\pi} \frac{\rd \theta_2}{2\pi} \,
e^{-i\theta_1 (N-1)} e^{i\theta_2 (N-1)} \\ &\hspace{3cm} \times
\left\langle W(e^{-i\theta_1} \sqrt{N} \ph) \Omega, \,
W(e^{-i\theta_2} \sqrt{N} \ph) \Omega \right\rangle \, .
\end{split}
\end{equation}
We introduce the notation
\begin{equation}
\begin{split}
f_N (x) = &d^2_N \int_0^{2\pi}
\frac{\rd \theta_1}{2\pi} \int_0^{2\pi} \frac{\rd \theta_2}{2\pi} \,
e^{-i\theta_1 (N-1)} e^{i\theta_2 N} \\ &\hspace{3cm} \times
\left\langle W(e^{-i\theta_1} \sqrt{N} \ph) \Omega, \left( a_x (t) -
e^{-i\theta_2} \sqrt{N} \ph_t (x) \right)W(e^{-i\theta_2} \sqrt{N}
\ph) \Omega \right\rangle\,.
\end{split}
\end{equation}
Since
\begin{equation}
\begin{split} d_N \int_0^{2\pi} \frac{\rd \theta}{2\pi} e^{i\theta (N-1)}
W (e^{-i\theta} \sqrt{N} \ph) \Omega &= d_N e^{-N/2}
\sum_{j=0}^{\infty} \left( \int_0^{2\pi} \frac{\rd \theta}{2\pi}
e^{i\theta (N-1-j)} \right) N^{j/2} \frac{(a^* (\ph))^j}{j!} \Omega
\\ &= d_N \frac{e^{-N/2} N^{(N-1)/2}}{\sqrt{N-1!}} \frac{(a^*
(\ph))^{N-1}}{N-1!} \Omega \\ &= \ph^{\otimes N-1} \, ,
\end{split}
\end{equation}
we obtain, from (\ref{eq:onepart}), that
\begin{equation}
\begin{split}
\gamma^{(1)}_{N,t} (x;y)= \; & \frac{d^2_N}{N} \int_0^{2\pi}
\frac{\rd \theta_1}{2\pi} \int_0^{2\pi} \frac{\rd \theta_2}{2\pi} \,
e^{-i\theta_1 N} e^{i\theta_2 N}  \left\langle W(e^{-i\theta_1}
\sqrt{N} \ph) \Omega, \, \left( a_y^* (t) - e^{i\theta_1} \sqrt{N}
\overline{\ph}_t (y) \right) \right. \\ &\hspace{3cm} \times \left.
\left( a_x (t) - e^{-i\theta_2} \sqrt{N} \ph_t (x) \right) \,
W(e^{-i\theta_2} \sqrt{N} \ph) \Omega \right\rangle \\ &+ \frac{\overline{\ph}_t (y) f_N (x)}{\sqrt{N}} + \frac{\ph_t (x) \overline{f_N} (y)}{\sqrt{N}} + \ph_t (x) \overline{\ph}_t (y)\,.
\end{split}
\end{equation}
Thus
\begin{equation}\label{eq:dif}
\begin{split}
\Big| \gamma^{(1)}_{N,t} (x;y) - \ph_t (x) \overline{\ph}_t (y)
\Big| \leq \; & \frac{d^2_N}{N} \,  \int_0^{2\pi} \frac{\rd
\theta_1}{2\pi} \int_0^{2\pi} \frac{\rd \theta_2}{2\pi} \; \Big\|
\left( a_y (t) - e^{-i\theta_1} \sqrt{N}
\ph_t (y) \right) W(e^{-i\theta_1} \sqrt{N} \ph) \Omega \Big\| \\
&\hspace{4cm} \times \Big\| \left( a_x (t) - e^{-i\theta_2} \sqrt{N}
\ph_t (x) \right) \, W(e^{-i\theta_2} \sqrt{N} \ph) \Omega \Big\| \\
& + \frac{|\ph_t (x)| |f_N (y)|}{\sqrt{N}} + \frac{|\ph_t (y)||f_N (x)|}{\sqrt{N}}  \\ \leq \; & \frac{d^2_N}{N} \, \int_0^{2\pi}
\frac{\rd \theta_1}{2\pi} \int_0^{2\pi} \frac{\rd \theta_2}{2\pi} \;
\| a_y \cU^{\theta_1}_{N} (t;0) \Omega \| \, \| a_x \cU^{\theta_2}_{N}
(t;0) \Omega \| \\ & + \frac{|\ph_t (x)| |f_N (y)|}{\sqrt{N}} + \frac{|\ph_t (y)||f_N (x)|}{\sqrt{N}}
\end{split}
\end{equation}
where the unitary evolutions $\cU^{\theta}_{N} (t;s)$ are defined as
in (\ref{eq:cUN}), but with $\ph_t$ replaced\footnote{We are making use here of the
important fact that if $\ph_t$ solves the nonlinear equation
(\ref{eq:hartree2}), then $e^{i\theta} \ph_t$ is also a solution of
the same equation, for any fixed real $\theta$.} by $e^{-i\theta} \ph_t$
in the generator (\ref{eq:cLN}). Taking the square
of (\ref{eq:dif}) and integrating over $x,y$, we obtain
\begin{equation}
\begin{split}
\int \rd x \rd y \; |\gamma^{(1)}_{N,t} (x;y) - \ph_t (x)
&\overline{\ph}_t (y) \Big|^2  \\ \leq \; & 2\frac{d^4_N}{N^2} \,
\int_0^{2\pi} \frac{\rd \theta_1}{2\pi} \int_0^{2\pi} \frac{\rd
\theta_2}{2\pi} \; \| \cN^{1/2} \cU^{\theta_1}_{N} (t;0) \Omega \|^2
\, \| \cN^{1/2} \cU^{\theta_2}_{N} (t;0) \Omega \|^2 \\ & + \frac{4}{N} \int \rd x |f_N (x)|^2
\end{split}
\end{equation}
Using Proposition \ref{prop:1cou} and the fact that $d_N \simeq N^{1/4}$ to control the first term, and using Lemma \ref{lm:improve} to control the second term on the r.h.s. of the last equation, we find constants $C,K$, only depending on $\| \ph \|_{H^1}$ and on
the constant $D$ in (\ref{eq:assump}) such that
\begin{equation}
\| \gamma^{(1)}_{N,t} - |\ph_t \rangle \langle \ph_t|\|_{\text{HS}}
\leq \frac{C}{N^{1/2}} \, \exp (K t)\,.
\end{equation}
This proves (\ref{eq:HS}) and thus concludes the proof of Theorem
\ref{thm:fact}.
\end{proof}

\begin{lemma}\label{lm:improve}
Let $\ph_t$ be a solution to the Hartree equation (\ref{eq:hartree}) with initial data $\ph \in H^1 (\bR^3)$ with $\| \ph \|=1$. Let
\begin{equation*}
\begin{split}
f_N (x) =& d^2_N \int_0^{2\pi}
\frac{\rd \theta_1}{2\pi} \int_0^{2\pi} \frac{\rd \theta_2}{2\pi} \,
e^{-i\theta_1 (N-1)} e^{i\theta_2 N} \\ &\hspace{3cm} \times
\left\langle W(e^{-i\theta_1} \sqrt{N} \ph) \Omega, \left( a_x (t) -
e^{-i\theta_2} \sqrt{N} \ph_t (x) \right)W(e^{-i\theta_2} \sqrt{N}
\ph) \Omega \right\rangle\,.
\end{split}
\end{equation*}
Then there exist constants $C,K$ (only depending on $\| \ph \|_{H^1}$ and on the constant $D$ in (\ref{eq:assump}) such that
\[ \int \rd x \, |f_N (x)|^2 \leq C e^{Kt}\,. \]
\end{lemma}

\begin{proof}
Using that
$$
\left( a_x (t) -
e^{-i\theta_2} \sqrt{N} \ph_t (x) \right)W(e^{-i\theta_2} \sqrt{N}\ph)=
 W(e^{-i\theta_2} \sqrt{N}\ph) {\cal U}_N^{\theta_2}(0;t)a_x  {\cal U}_N^{\theta_2}(t;0)
$$
where the unitary evolution $\cU^\theta_N (t;s)$ is defined as
in (\ref{eq:cUN}), but with $\ph_t$ replaced by $e^{-i\theta} \ph_t$
in the generator (\ref{eq:cLN}), we can rewrite $f_N (x)$ as
\begin{equation}\label{eq:fN1}
\begin{split}
f_N (x) = \int_0^{2\pi}
\frac{\rd \theta_2}{2\pi} \left\langle \psi (\theta_2) ,
{\cal U}_N^{\theta_2}(0;t)a_x  {\cal U}_N^{\theta_2}(t;0)\Omega \right\rangle
\end{split}
\end{equation}
with \begin{equation}\label{eq:psi21} \psi (\theta_2) =  d_N^2 \int_0^{2\pi}
\frac{\rd \theta_1}{2\pi} \,
e^{i\theta_1 (N-1)} e^{-i\theta_2 N}  W^*(e^{-i\theta_2} \sqrt{N}\ph) W(e^{-i\theta_1} \sqrt{N} \ph) \Omega\, . \end{equation}
Performing the integration over $\theta_1$, we immediately obtain
\begin{equation}\label{eq:psi22}
\psi (\theta_2)= d_N  \, e^{-i\theta_2 N}  W^*(e^{-i\theta_2} \sqrt{N}\ph) \ph^{\otimes (N-1)}\,.
\end{equation}
It is also possible to expand $\psi (\theta_2)$ in a sum of factors living in the different sectors of the Fock space. {F}rom Eq. (\ref{eq:coh}) and Lemma \ref{lm:coh}, we compute
\begin{equation}
\begin{split}
W^*(e^{-i\theta_2} \sqrt{N}\ph) W(e^{-i\theta_1} \sqrt{N} \ph) \Omega &= W(-e^{-i\theta_2} \sqrt{N}\ph) W(e^{-i\theta_1} \sqrt{N} \ph) \Omega \\ &= e^{iN \text{Im} e^{i(\theta_2-\theta_1)}} W ((e^{-i\theta_1} - e^{-i\theta_2}) \sqrt{N} \ph ) \Omega \\ &=
e^{-N} e^{N e^{i(\theta_2-\theta_1)}} \sum_{m \geq 0} \frac{N^{m/2} (e^{-i\theta_1} - e^{-i\theta_2})^m}{\sqrt{m!}} \ph^{\otimes m}
\end{split}
\end{equation}
which implies (using the periodicity in the variable $\theta_1$)
\begin{equation*}
\begin{split}
\psi (\theta_2) = \; & d_N^2 e^{-N}
\sum_{m=0}^{\infty} \frac{N^{m/2}}{\sqrt{m!}}  \int_0^{2\pi}
\frac{\rd \theta_1}{2\pi} \,
e^{i\theta_1 (N-1)} e^{-i\theta_2 (m+1)}  e^{N e^{-i\theta_1}} (e^{-i\theta_1} - 1)^m  \ph^{\otimes m}\,.
\end{split}
\end{equation*}
Switching to the complex variable $z=e^{-i\theta_1}$ we obtain
\begin{equation*}
\begin{split}
\psi (\theta_2) = \; & - d_N^2 e^{-N} \sum_{m \geq 0} \frac{N^{m/2}}{\sqrt{m!}}
e^{-i\theta_2 (m+1)} \int \frac{\rd z}{2\pi i} \, z^{-N}  e^{Nz} (z - 1)^m \ph^{\otimes m}
\end{split}
\end{equation*}
where the $z$ integral is over the circle of radius one around the origin (in clock-wise sense). Changing variables $z\to Nz$, and using that $d_N = e^{N/2} \sqrt{N!}/ N^{N/2}$, we obtain
\begin{equation}\label{eq:psi2}
\begin{split}
\psi (\theta_2) = \; &- (N-1)! \sum_{m=0}^\infty \frac{N^{-\frac{m}{2}}}{\sqrt{m!}}  e^{-i\theta_2 (m+1)}  \int \frac{\rd z}{2\pi i} \, z^{-N} e^{z} (z-N)^m \ph^{\otimes m} \\
= \; &\sum_{m=0}^\infty \frac{N^{-\frac m2}}{\sqrt{m!}} \cR_m \,
e^{-i\theta_2(m+1)} \ph^{\otimes m}
\end{split}
\end{equation}
where we defined
\begin{equation}\label{eq:Rm}
{\cal R}_m=\frac {\rd^{N-1}}{\rd z^{N-1}}\left (e^z (z-N)^m\right)|_{z=0}\,.
\end{equation}
Comparing (\ref{eq:psi2}) with (\ref{eq:psi22}), we obtain the identity
\begin{equation}\label{eq:parse}
\sum_{m=0}^\infty \frac{{\cal R}^2_m}{N^{m}{m!}}=d_N^2 \, .
\end{equation}
It is also possible to obtain pointwise bounds on the coefficients $\cR_m$.
{F}rom \eqref{eq:Rm} we deduce that for $m\le (N-1)$
\begin{equation}\label{eq:repR}
{\cal R}_m = \sum_{k=0}^m (-1)^{m-k} \frac {(N-1)! m! N^{m-k}}{k! (N-1-k)! (m-k)!}
=\sum_{k=0}^m (-1)^{m-k} N^{m-k} (N-1)...(N-k) \frac {m!}{k! (m-k)!}\,.
\end{equation}
The coefficients ${\mathcal R}_m$ turn out to be intimately connected with the classical
system of orthogonal Laguerre polynomials. Recall that the associated Laguerre polynomial
$L_n^{(\alpha)}(x)$ admits the following representation
$$
L_n^{(\alpha)}(x)=\sum_{k=0}^n (-1)^k \frac {(n+\alpha)!}{k! (n-k)! (\alpha+k)!} x^k.
$$
Therefore
$$
{\cal R}_m = (-1)^m m!\, L_m^{(N-m-1)}(N),
$$
which, for $N>m+1$, involves the value of the Laguerre polynomial $L_n^{(\alpha)}(N)$ with a  {\it positive} index $\alpha$. Asymptotic expansions and estimates for the Laguerre polynomials is
a classical subject, see \cite{Szego} and references therein.
However for the indices $\alpha=N-m-1$, $n=m$ with $N \gg m$ the value of $x=N$ belongs to the oscillatory regime of the behavior of $L_n^{(\alpha)}(x)$ and the sharp estimates for those values of parameters have been only obtained recently in \cite{Krasikov}, where it is proven that, for $\alpha>-1$, $n\ge 2$ and the values of $x\in (q^2,s^2)$ the function $L^{(\alpha)}_n(x)$ obeys the bound
$$
|L^{(\alpha)}_n(x)|< \sqrt{\frac {(n+\alpha)!}{n!}}\sqrt{\frac {x(s^2-q^2)}{r(x)}} e^{\frac x2} x^{-\frac {\alpha+1}2},
$$
where
$$
s=(n+\alpha+1)^{\frac 12} + n^{\frac 12}, \quad q=(n+\alpha+1)^{\frac 12} - n^{\frac 12},\quad
r(x)=(x-q^2)(s^2-x).
$$
As a consequence, we obtain that
$$
|L_m^{(N-m-1)}(N)|<  \sqrt{\frac {(N-1)!}{m!}}\sqrt{\frac {4 N\sqrt{Nm}}{4Nm-m^2}} e^{\frac N2} N^{-\frac {N-m}2}
$$
Assuming that $m\le N$ and using the asymptotics $(N-1)!\sim N^{N-1/2} e^{-N}$ we obtain
$$
|L_m^{(N-m-1)}(N)|\lesssim m^{-\frac 14} (m!)^{-\frac 12} N^{\frac m2}
$$
and therefore
$$
\frac {{\cal R}_m}{(m!)^{\frac 12} N^{\frac m2}}\lesssim m^{-\frac 14}\,.
$$
Summarizing, the coefficients $\cA_m = \cR_m / (m!^{1/2} N^{m/2})$ appearing in the expansion (\ref{eq:psi2}) of $\psi (\theta_2)$ satisfy the bounds
\begin{equation}\label{eq:Am}
\begin{split}
|\cA_m| &\leq C m^{-1/4} \qquad \text{for all } m \leq N \quad \text{and} \\
\sum_{m=0}^{\infty} \cA_m^2 &= d_N^2 \leq C N^{1/2} \,.
\end{split}
\end{equation}

Inserting (\ref{eq:psi2}) into (\ref{eq:fN1}) we obtain
\begin{equation}\label{eq:fN2}
f_N (x) =  \sum_{m=0}^\infty \cA_m \,
\int_0^{2\pi} \frac{\rd \theta}{2\pi} \, e^{i\theta(m+1)} \left\langle
\ph^{\otimes m},
{\cal U}_N^{\theta}(0;t) a_x  {\cal U}_N^{\theta}(t;0)\Omega \right\rangle \end{equation}
and therefore
\begin{equation}
\begin{split}
|f_N (x)| = &\; \int_0^{2\pi} \frac{\rd \theta}{2\pi} \, \sum_{m=0}^\infty \frac{|{\cal A}_m|}{\sqrt{m+1}} \left| \left\langle \ph^{\otimes m},
(\cN+1)^{1/2} \,{\cal U}_N^{\theta}(0;t) a_x {\cal U}_N^{\theta}(t;0)\Omega \right\rangle\right| \\ \leq & \, \left(\sum_{m=0}^\infty \frac{|{\cal A}_m|^2}{m+1} \right)^{1/2} \int_0^{2\pi} \frac{\rd \theta}{2\pi}  \left\|
(\cN+1)^{1/2} \,{\cal U}_N^{\theta}(0;t) a_x {\cal U}_N^{\theta}(t;0)\Omega \right\|\,.
\end{split}
\end{equation}
{F}rom (\ref{eq:Am}), we obtain
\begin{equation}
\sum_{m=0}^\infty \frac{|{\cal A}_m|^2}{m+1}  \leq C \sum_{m=0}^{N-1} \frac{1}{(m+1)^{3/2}} + \frac{1}{N} \sum_{m \geq N} |A_m|^2 \leq \const \,.
\end{equation}
On the other hand, from Proposition \ref{prop:1cou}, we have \begin{equation*}
\begin{split}
\left\| (\cN+1)^{1/2} \,{\cal U}_N^{\theta}(0;t) a_x {\cal U}_N^{\theta}(t;0)\Omega \right\|^2 &\leq C e^{\wt{K} t} \left\| (\cN+1)^2 a_x {\cal U}_N^{\theta}(t;0)\Omega \right\|^2 \leq C e^{\wt{K} t} \left\| a_x \cN^2 {\cal U}_N^{\theta}(t;0)\Omega \right\|^2.
\end{split}
\end{equation*}
Thus, applying once more Proposition \ref{prop:1cou}, we find
\[ \int \rd x \, |f_N (x)|^2 \leq C e^{\wt{K}t} \int_0^{2\pi} \frac{\rd \theta_2}{2\pi} \langle {\cal U}_N^{\theta}(t;0)\Omega, \cN^5 {\cal U}_N^{\theta}(t;0)\Omega  \rangle \leq C e^{Kt} \,.\]
\end{proof}

\bigskip

{\it Acknowledgements.} B. Schlein is grateful to L. Erd\H os and
H.-T. Yau for many stimulating discussions concerning the
dynamics of mean field systems.

\end{document}